\documentclass[11pt]{article}
\usepackage[T1]{fontenc}
\usepackage[margin = 1in]{geometry}
\usepackage{textcomp}
\usepackage{graphicx}

\usepackage{amsmath,amssymb,amsthm,hyperref,color}
\usepackage{float}
\usepackage[font=small]{caption}
\usepackage[font=footnotesize]{subcaption}
\usepackage[inline]{enumitem}
\usepackage{filecontents}
\usepackage{tikz}
\usepackage{pgfplots}
\usetikzlibrary{arrows,shapes.geometric,calc,decorations,backgrounds}
\hypersetup{colorlinks=true,citecolor=blue, linkcolor=blue,
urlcolor=blue}

\newtheorem{theorem}{Theorem}
\newtheorem{lemma}[theorem]{Lemma}

\newtheorem{remark}[theorem]{Remark}
\newtheorem{example}[theorem]{Example}

\newtheorem{definition}[theorem]{Definition}

\def\calI{\mathcal{I}}
\def\Gc{\mathcal{G}}
\def\Tc{\mathcal{T}}
\def\Yc{\mathcal{Y}}
\def\Ac{\mathcal{A}}
\def\qb{\mathbf{q}}
\def\Ising{\mathsf{Ising}}
\def\AvgCut{\textsc{Avg-Cut}}
\def\MaxCut{\textsc{Max-Cut}}
\def\Cut{\mathrm{Cut}}
\def\bit{\mathrm{bits}}
\newcommand{\pl}{\textup{\texttt{+}}}
\newcommand{\mi}{\textup{\texttt{-}}}
\newcommand{\plm}{\textup{\texttt{\textpm}}}
\newcommand{\mip}{\textup{\texttt{\raisebox{.1em}{%
	\reflectbox{\rotatebox[origin=c]{180}{\textpm}}}}}}

\def\Zin{Z^{\mathsf{in}}}
\def\Zout{Z^{\mathsf{out}}}

\def\Lb{\mathbf{L}}

\def\lambdab{\ensuremath{\boldsymbol{\lambda}}}

\def\Mc{\ensuremath{\mathcal{M}}}
\def\Eb{\ensuremath{\mathbf{E}}}

\def\Uc{\ensuremath{\mathcal{U}}}
\def\Dc{\ensuremath{\mathcal{D}}}
\def\T{\ensuremath{\intercal}}
\newcommand{\EST}{\mathsf{EST}}

\newcommand{\fptas}{\mathsf{FPTAS}}
\newcommand{\fpras}{\mathsf{FPRAS}}

\newcommand{\eps}{\epsilon}

\makeatletter
\def\prob#1#2#3{\goodbreak\begin{list}{}{\labelwidth\z@ \itemindent-\leftmargin
                        \itemsep\z@  \topsep6\p@\@plus6\p@
                        \let\makelabel\descriptionlabel}
                \item[\it Name]#1
               \item[\it Instance]                #2
                \item[\it Output]#3
                \end{list}}
\makeatother

\title{The complexity of approximating averages on bounded-degree graphs}
\author{
Andreas Galanis\thanks{
  Department of Computer Science, University of Oxford, Wolfson Building, Parks Road, Oxford, OX1~3QD, UK.}
\and
Daniel \v{S}tefankovi\v{c}\thanks{Department of Computer Science, University of Rochester, USA. Research supported in part by NSF grant CCF-1563757.}
\and
Eric Vigoda\thanks{Department of Computer Science, University of California, Santa Barbara, CA 93106, USA. Research supported in part by NSF grant CCF-1563838.}
}
\date{July 19, 2020}

\begin{document}

\maketitle
\thispagestyle{empty}
\begin{abstract} 
We prove that, unless P$=$NP, there is no polynomial-time algorithm to approximate within
some multiplicative constant the average size of an independent set in graphs of maximum degree~$6$.
This is a special case of a more general result for the hard-core model defined on 
independent sets weighted by a parameter $\lambda>0$.     In the general setting, we prove that, unless P$=$NP,
for all $\Delta\geq 3$, all $\lambda>\lambda_c(\Delta)$,
there is no FPTAS which applies to all graphs of maximum degree $\Delta$ for
computing the average size of the independent set in the Gibbs distribution,
where $\lambda_c(\Delta)$ is the critical point for the uniqueness/non-uniqueness phase transition
on the $\Delta$-regular tree.  Moreover, we prove that for $\lambda$  in a dense set of this
non-uniqueness region the problem is NP-hard to approximate within some constant factor.
Our work extends to the antiferromagnetic Ising model and generalizes to all 2-spin antiferromagnetic models,
establishing hardness of computing the average magnetization in the tree non-uniqueness region.

Previously, Schulman, Sinclair and Srivastava (2015) showed that it is \#P-hard to compute the average magnetization exactly, but no hardness of approximation
results were known. Hardness results of Sly (2010) and Sly and Sun (2014) for  approximating the partition function do not imply hardness of computing averages.
The new ingredient in our reduction is an intricate construction of pairs of rooted trees whose marginal distributions at the root agree but their derivatives disagree. The main technical contribution is controlling what marginal distributions and derivatives are achievable and using Cauchy's functional equation to argue existence of the gadgets.
\end{abstract}

\newpage

\clearpage
 \setcounter{page}{1}
\section{Introduction}

This paper addresses the computational problem of computing averages of simple functions over 
combinatorial structures of a graph.  Can we estimate elementary statistics of combinatorial structures in polynomial-time?  This genre of problems is nicely illustrated for the example of independent sets.

Given a graph $G=(V,E)$ can we efficiently estimate the average size of an independent set in~$G$?
For graph $G=(V,E)$, let $\calI_G$ denote the set of independent sets of $G$, and let
$\mu_G$ denote the uniform distribution over $\calI_G$.  Denote the average independent set size by $\Mc(G)=\Eb_{\sigma\sim \mu}\big[|\sigma|\big]$.

Schulman \emph{et al.}~\cite{SSS} (see also~\cite{SinclairSrivastava}) proved that exactly computing the average independent set size is \#P-hard for bounded-degree  graphs.  
We investigate approximation algorithms for the problem.  For a constant $C>1$ we say there is a $C$-approximation algorithm for
the average independent set size if the algorithm outputs an estimate $\EST$ for which $\tfrac{1}{C}\times\Mc(G)\leq\EST\leq C\times \Mc(G)$.
An $\fptas$ is an 
algorithm which, for any input $G=(V,E)$ and $\eps>0$, achieves a $(1+\eps)$-approximation factor in time $poly(|V|,1/\eps)$.

Weitz~\cite{Weitz} presented an $\fptas$ for estimating $|\calI_G|$, the number of independent sets,
in graphs of maximum degree $\leq 5$, and this immediately yields an $\fptas$ for the average independent set size in graphs of maximum degree $\leq 5$, see also \cite{ALO20,PR,Peters} for new algorithmic approaches. 
We prove that this result is optimal. In fact, we prove that approximating the average independent set size in graphs of maximum degree $6$ is
hard within a constant factor.

\begin{theorem}\label{thm:mainind1}
There is a constant $C>1$ such that, for all integers $\Delta\geq 6$, for  graphs $G$ of maximum degree $\Delta$ there is no polynomial-time $C$-approximation algorithm for computing the average independent-set size of $G$, unless \emph{P=NP}.
\end{theorem}

This theorem is a special case of a more general result for the hard-core model, which is a statistical physics model of
particular combinatorial interest.  The hard-core model is defined on independent sets weighted by a
parameter $\lambda>0$, known as the fugacity.  An independent set $\sigma\in\calI_G$ has weight $w(\sigma) = \lambda^{|\sigma|}$.
For a graph $G$ and fugacity $\lambda>0$, the Gibbs distribution is defined as
$\mu_{G;\lambda}(\sigma)=w(\sigma)/Z_{G;\lambda}$ where the normalizing factor $Z_{G;\lambda} = \sum_{\tau\in\calI_G} w(\tau)$ is known as the partition function.

The earlier case of unweighted independent sets corresponds to the hard-core model with $\lambda=1$.  
Hence we use the same notation $\Mc(G)=\Eb_{\sigma\sim \mu}\big[|\sigma|\big]$ to denote
the average independent set size in the Gibbs distribution.

On the $\Delta$-regular tree, the hard-core model undergoes a phase transition at the critical point 
$\lambda_c(\Delta)=\tfrac{(\Delta-1)^{\Delta-1}}{(\Delta-2)^{\Delta}}$.   When $\lambda\leq\lambda_c(\Delta)$ there is a unique
infinite-volume Gibbs measure on the $\Delta$-regular tree (roughly, this corresponds to the decay of the ``influence'' of the leaves on the root),
whereas when $\lambda>\lambda_c(\Delta)$ there is non-uniqueness, i.e., there are multiple infinite-volume Gibbs measures.

There is an interesting computational phase transition for graphs of maximum degree $\Delta$ that occurs at this same tree threshold.
For all constant $\Delta$, all $\lambda<\lambda_c(\Delta)$, all graphs of maximum degree $\Delta$,
there is an $\fptas$ for the partition function~\cite{Weitz}.   On the other side, for all $\Delta\geq 3$, all $\lambda>\lambda_c(\Delta)$,
there is no $\fpras$ for approximating the partition function on graphs of maximum degree $\Delta$, unless NP$=$RP~\cite{Sly10,SlySun,GSVIsing}.
However, hardness of computing partition functions does not imply hardness of computing averages in the Gibbs distribution; see the case of  the antiferromagnetic Ising model with no external field discussed in Section~\ref{sec:Ising}.

We prove that computing the average independent set size also undergoes a computational phase transition at
the tree critical point $\lambda_c(\Delta)$.  As before, Weitz's algorithmic result~\cite{Weitz}  
yields, for all constant~$\Delta$, all~$\lambda<\lambda_c(\Delta)$, all graphs~$G$ of maximum degree~$\Delta$, 
an $\fptas$ for the average independent set size $\Mc(G)$ (more generally, an approximate sampling algorithm implies an algorithm for approximating averages).   We prove that this result is optimal: when $\lambda>\lambda_c(\Delta)$ there
is no $\fptas$ for the average independent set size.

\begin{theorem}\label{thm:mainind2}
Let $\Delta\geq 3$ be an integer and $\lambda>\lambda_c(\Delta)$. Then, for graphs $G$ of maximum degree $\Delta$, there is no $\fptas$ for computing the average independent-set size in the hard-core distribution $\mu_{G;\lambda}$, unless \emph{P=NP}. 
\end{theorem}
In fact, our inapproximability result for general $\lambda>\lambda_c(\Delta)$ is stronger, it actually precludes approximation algorithms with factors of $1\pm \tfrac{\delta}{\log n}$, where $n$ is the number of the vertices of the input graph and $\delta$ is an appropriate constant, see Theorem~\ref{thm:maingen2} below for the precise statement. For a dense set of $\lambda$, we actually obtain constant-factor inapproximability (analogously to Theorem~\ref{thm:mainind1}).
\begin{theorem}\label{thm:mainind3}
Let $\Delta\geq 3$ be an integer. Then, for every real $\lambda>\lambda_c(\Delta)$ and $\epsilon>0$, there is an algebraic number $\widehat{\lambda}$ with $|\widehat{\lambda}-\lambda|\leq \epsilon$ and  a constant $C=C(\widehat{\lambda},\Delta)>1$ such that, for graphs $G$ of maximum degree $\Delta$, there is no poly-time $C$-approximation algorithm for computing the average independent-set size of $G$ in the hard-core distribution $\mu_{G;\widehat{\lambda}}$, unless \emph{P=NP}.
\end{theorem}

\subsection{Results for the antiferromagnetic Ising model}\label{sec:Ising}
Our results extend to the antiferromagnetic Ising model.  Let $G=(V,E)$ be a graph. For $\beta,\lambda>0$, let $\mu_{G;\beta,\lambda}$ denote the Ising distribution on $G$ with edge activity $\beta$ and external field $\lambda$, i.e., for $\sigma:V\rightarrow \{0,1\}$ we have 
\[\mu_{G;\beta,\lambda}(\sigma)=\frac{\lambda^{|\sigma|}\beta^{m(\sigma)}}{Z_{G;\beta,\lambda}},\]
where $m(\sigma)$ denotes the number of monochromatic edges in $G$ under $\sigma$, i.e., edges whose endpoints have the same spin. The model is called antiferromagnetic if $\beta\in(0,1)$ and ferromagnetic, otherwise. We define the average magnetization of $G$ to be the average number of vertices with spin 1, i.e., \[\Mc_{\beta,\lambda}(G)=\frac{1}{Z_{G;\beta,\lambda}}\sum_{\sigma:V\rightarrow \{0,1\}}|\sigma|\lambda^{|\sigma|}\beta^{m(\sigma)}=\Eb_{\sigma\sim \mu}\big[|\sigma|\big],\]
where for a configuration $\sigma:V\rightarrow \{0,1\}$, we use $|\sigma|$ to denote $\sum_{v\in V(G)} \sigma(v)$, i.e., the total number of vertices with spin $1$.

In the ferromagnetic case, there is an $\mathsf{FPRAS}$ for approximating the magnetization for all $\beta>1$ and $\lambda>0$, due to the algorithm of Jerrum and Sinclair \cite{JS:ising}.  For $\Delta\geq 3$, let $\beta_c(\Delta)=\tfrac{\Delta-2}{\Delta}$. It is known that the antiferromagnetic Ising model with edge activity $\beta$ and external field $\lambda$ has non-uniqueness on the infinite $\Delta$-regular tree iff $\beta\in (0,\beta_c)$ and $\lambda \in (1/\lambda^{\Ising}_c,\lambda^{\Ising}_c)$ for some explicit $\lambda^{\Ising}_c=\lambda^{\Ising}(\beta,\Delta)>1$.  For all constant $\Delta$, in the tree
uniqueness region there is an $\fptas$ for the partition function for graphs of maximum degree~$\Delta$~\cite{SST,LLY}, and, once again,
this implies an $\fptas$ for the magnetization.  In the tree non-uniqueness region, for all $\Delta\geq 3$, unless NP$=$RP  
there is no $\fpras$ for the partition function for graphs of maximum degree~$\Delta$~\cite{SlySun}. We prove that
approximating the magnetization is also intractable in the tree non-uniqueness region, apart from the case $\lambda=1$ (where  the magnetization can be computed trivially for all graphs $G$ since it equals $\tfrac{1}{2}|V(G)|$). 
\begin{theorem}\label{thm:mainIs2}
Let $\Delta\geq 3$ be an integer, $\beta\in (0,\beta_c(\Delta))$ and  $\lambda\in (\tfrac{1}{\lambda_c},\lambda_c)$ with $\lambda\neq 1$, where $\lambda_c=\lambda^{\Ising}_c(\Delta,\beta)$. Then, for graphs $G$ of maximum degree $\Delta$, there is no $\fptas$ for computing the average magnetization in the Ising distribution $\mu_{G;\beta, \lambda}$, unless \emph{P=NP}. 
\end{theorem}
As in Theorem~\ref{thm:mainind2}, the inapproximability factor in Theorem~\ref{thm:mainIs2} is in fact stronger, see Theorem~\ref{thm:maingen2} below for the precise statement. For a dense set of $\lambda$, we again obtain constant-factor inapproximability.
\begin{theorem}\label{thm:mainIs3}
Let $\Delta\geq 3$ be an integer, rational $\beta\in (0,\beta_c(\Delta))$ and  $\lambda_c=\lambda^{\Ising}_c(\Delta,\beta)$. Then, for every $\lambda\in  (\tfrac{1}{\lambda_c},\lambda_c)$ and $\epsilon>0$, there is an algebraic number $\widehat{\lambda}$ with $|\widehat{\lambda}-\lambda|\leq \epsilon$ and  a constant $C=C(\beta,\widehat{\lambda},\Delta)>1$ such that, for graphs $G$ of maximum degree $\Delta$, there is no poly-time $C$-approximation algorithm for computing the average magnetization $\Mc_{\beta,\hat{\lambda}}(G)$, unless \emph{P=NP}.
\end{theorem}

\subsection{Results for general antiferromagnetic 2-spin systems}
While the hard-core model and the Ising model are the most canonical 2-spin models, the results of the previous two sections will be obtained as special cases of the following results for general antiferromagnetic 2-spin systems. This more general perspective will also allow us to give a unified proof of Theorems~\ref{thm:mainind2},~\ref{thm:mainind3},~\ref{thm:mainIs2},~\ref{thm:mainIs3}.

Let $G=(V,E)$ be a graph. For $\beta,\gamma,\lambda>0$, let $\mu_{G;\beta,\gamma,\lambda}$ denote the Gibbs distribution on $G$ with edge activities $\beta,\gamma$ and external field $\lambda$, i.e., for $\sigma:V\rightarrow \{0,1\}$ we have 
\[\mu_{G;\beta,\lambda}(\sigma)=\frac{\lambda^{|\sigma|}\beta^{m_{0}(\sigma)}\gamma^{m_{1}(\sigma)}}{Z_{G;\beta,\gamma,\lambda}},\]
where $m_{0}(\sigma), m_{1}(\sigma)$ denotes the number of edges in $G$ whose endpoints are assigned under $\sigma$ the pair of spins $(0,0)$ and $(1,1)$, respectively. 

The parameter pair $(\beta,\gamma)$ is called \emph{antiferromagnetic} if $\beta\gamma\in[0,1)$ and at least one of $\beta,\gamma$ is non-zero, and it is called ferromagnetic, otherwise. Note that the hard-core model is the case $\beta=1,\gamma=0$ (under the convention that $0^0\equiv 1$) whereas the antiferromagnetic Ising model is the case $0<\beta=\gamma<1$.

We next define the range of parameters $(\beta,\gamma,\lambda)$ where our inapproximability results for the magnetization apply; these are precisely the parameters where the spin system exhibits non-uniqueness on the infinite $\Delta$-regular tree, apart from the case of the antiferromagnetic Ising model with $\lambda=1$, where  as noted in Section~\ref{sec:Ising} the magnetization can be computed trivially for all graphs $G$.
\begin{definition}\label{def:nonuniq}
Let $\Delta\geq 3$ be an integer. We let $\Uc_{\Delta}$ be the set of $(\beta,\gamma,\lambda)$ such that $(\beta,\gamma)$ is antiferromagnetic, and the (unique) fixpoint $x^*>0$ of the function $f(x)=\frac{1}{\lambda} \big(\frac{ \beta x+1}{x+\gamma}\big)^{\Delta-1}$ satisfies $|f'(x^*)|>1$.

We let $\Uc_{\Delta}^{*}=\Uc_{\Delta}\backslash \bigcup_{\beta\in (0,1)}\{(\beta,\beta,1)\}$ be the set of parameters in $\Uc_\Delta$ other than those where computing the magnetization is trivial. 
\end{definition}
We  note that Li, Lu, Yin \cite{LLY} define a notion of ``up-to-$\Delta$'' uniqueness which requires $(\beta,\gamma,\lambda)$ to  be in uniqueness for every $d\leq \Delta$; they obtain an $\fptas$ for the partition function in that region. The complement of their region corresponds to non-uniqueness in the sense of Definition~\ref{def:nonuniq} for some $d\leq \Delta$. Our inapproximability results extend to this bigger region by applying our theorems for smaller values of $\Delta$.

Our first inapproximability result for general antiferromagnetic 2-spin models, which is a generalization/strengthening of Theorems~\ref{thm:mainind2} and~\ref{thm:mainIs2}, is the following. The proof is given in Section~\ref{sec:maingen2}.
\begin{theorem}\label{thm:maingen2}
Let $\Delta\geq 3$ be an integer and $(\beta,\gamma,\lambda)\in \Uc_{\Delta}^*$. Then, for graphs $G\in \mathtt{G}_\Delta$, there is no $\fptas$ for computing the average magnetization in the Gibbs distribution $\mu_{G;\beta, \gamma, \lambda}$, unless P=NP. In fact, there is a constant $\kappa=\kappa(\Delta, \beta,\gamma, \lambda)>0$ such that there is no poly-time $\big(1+ \tfrac{\kappa}{\log n}\big)$-approximation algorithm for computing the average magnetization $\Mc_{\beta,\gamma,\lambda}(G)$, where $n=|V(G)|$.
\end{theorem}
We remark that a ``no-FPTAS'' result can be strengthened to an inapproximability factor of $\big(1\pm \tfrac{1}{n^{\epsilon}}\big)$ for any constant $\epsilon>0$ via standard powering techniques; the tighter hardness factor of $\big(1\pm \tfrac{\kappa}{\log n}\big)$ in Theorem~\ref{thm:maingen2} requires a substantially more delicate argument. 
\begin{theorem}\label{thm:maingen3}
Let $\Delta\geq 3$ be an integer and $(\beta,\gamma,\lambda)\in \Uc_{\Delta}$ with $\beta,\gamma$ rational numbers.   Then, for every $\epsilon>0$, there is an algebraic number $\widehat{\lambda}$ with $|\widehat{\lambda}-\lambda|\leq \epsilon$ and  a constant $C=C(\beta,\gamma,\widehat{\lambda},\Delta)>1$  such that, for graphs $G$ of maximum degree $\Delta$, there is no poly-time $C$-approximation algorithm for computing the average magnetization $\Mc_{\beta,\gamma,\hat{\lambda}}(G)$, unless \emph{P=NP}.
\end{theorem}

\section{Proof Outline}
In this section, we give some of the key elements of the techniques needed to prove our inapproximability results and conclude with the proof of Theorem~\ref{thm:mainind1}. We start by describing ``field gadgets'' and state the main lemmas that we will use in our reduction.

\subsection{Gadgets with approximately equal effective fields and different averages}\label{sec:f434}
For a graph $G$ and $\sigma:V\rightarrow \{0,1\}$, we say that a vertex $v\in V$ is occupied if $\sigma(v)=1$, and unoccupied otherwise.\footnote{This terminology is standard for the hard-core model, but it will be convenient to use it for general 2-spin systems.} For a subset $S\subseteq V$, we use $\sigma_S$ to denote the configuration on $S$ which is restriction of $\sigma$ on $S$.
\begin{definition}\label{def:field}
A {\em field gadget} is a rooted tree $\mathcal{T}$ whose root $\rho$ has degree one. For antiferromagnetic $(\beta,\gamma)$ and $\lambda>0$, let $\mu=\mu_{\Tc;\beta,\gamma,\lambda}$.
\begin{enumerate}
\item The {\em effective field} of the gadget, denoted by $R_{\mathcal{T}}=R_{\mathcal{T}}(\beta,\gamma,\lambda)$, is $1/\lambda$ times the ratio of the weight of
configurations where the root is occupied to the weight of configurations where the root is
unoccupied, i.e., $R_{\mathcal{T}}=\tfrac{1}{\lambda}\frac{\mu(\sigma_\rho=1)}{\mu(\sigma_\rho=0)}$. (The division by $\lambda$ is to avoid double-counting the contribution of the root later on.)
\item  The {\em magnetization gap} of the gadget, denoted by $M_{\mathcal{T}}=M_{\mathcal{T}}(\beta,\gamma,\lambda)$, is the expected number of occupied vertices conditioned on the root being occupied minus the expected number of occupied vertices conditioned on the root being unoccupied (the root is included in the count), i.e., 
\[M_{\mathcal{T}}=\Eb_{\sigma\sim \mu}[\,|\sigma|\mid \sigma(\rho)=1]-\Eb_{\sigma\sim\mu}[\,|\sigma|\mid \sigma(\rho)=0].\]
\end{enumerate}
In the special case that $\lambda=\frac{1-\beta}{1-\gamma}$ with $\beta\neq \gamma$, a field gadget consists of a rooted bipartite graph obtained from a rooted tree where some  of the leaves have been replaced by a distinct cycle of length four (by identifying the leaf with a vertex of the cycle). 
\end{definition}
Note that, in the case that $\lambda=\frac{1-\beta}{1-\gamma}$, the effective field of any tree gadget can be shown to be equal to 1 and the magnetization gap to 0, so that is why we need to consider the ``mildly non-tree'' construction in Definition~\ref{def:field}.

It will be useful to illustrate Definition~\ref{def:field} with a few examples.
\begin{example}\label{eee}
For example the rooted tree with one edge has effective field $R=\frac{1+\gamma\lambda}{\beta+\lambda}$.
The degenerate example where the root has degree zero has effective field $R=1$.
\end{example}

The following more interesting example in the case of independent sets will be used to prove Theorem~\ref{thm:mainind1}; it gives a pair of field gadgets with the same effective field but different magnetization gaps.
\begin{example}\label{e4f334}
Consider the independent set model (with $\lambda=1$). Consider the trees $\Tc_1,\Tc_2$ with roots $\rho_1,\rho_2$ below.
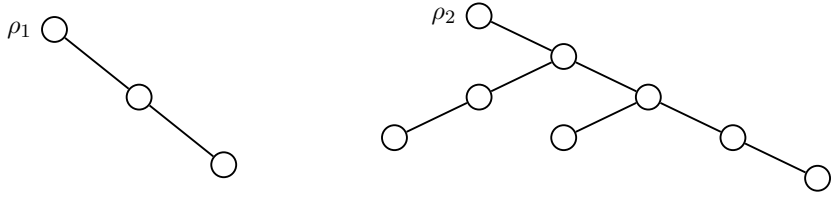
\begin{center}
\begin{minipage}{0.3\textwidth}
	\scalebox{0.9}[0.9]{
\begin{tikzpicture}[
	 level distance=1cm, sibling distance=2.5cm,
	rtn/.style = {shape=circle, draw=black, fill=white,  align=center, thick, minimum size=0.2cm}, 
	edge from parent/.style={draw,thick,-}]
	\node[rtn,label=left:$\rho_1$] {}
child [missing]
child	{ node[rtn] {} {
					child [ missing ]
						child { node[rtn] {} }
           }
       };
\end{tikzpicture}}
\end{minipage}
\begin{minipage}{0.5\textwidth}
	\scalebox{0.9}[0.9]{	
	\begin{tikzpicture}[
	 level distance=0.6cm, sibling distance=2.5cm,
	rtn/.style = {shape=circle, draw=black, fill=white,  align=center, thick, minimum size=0.2cm}, 
	edge from parent/.style={draw,thick,-}]
\node[rtn, label=left:$\rho_2$] {}
child [missing]
child	{ node[rtn] {} {
						child { node[rtn] {} {
										child	{ node[rtn] {} }
										child [ missing ]
											}
									}
						child { node[rtn] {} {
										child { node[rtn] {} }
										child	{ node[rtn] {} {
															child [ missing ]
															child { node[rtn] {} }
															}
													}
										}
						}
           }
       };
\end{tikzpicture}}
\end{minipage}
\end{center}
 Then $R_{\Tc_1}=R_{\Tc_2}$ since $R_{\Tc_1}=\frac{2}{3}$ and $R_{\Tc_2}=\tfrac{24}{36}$,  but $M_{\Tc_1}\neq M_{\Tc_2}$ since $M_{\Tc_1}=\tfrac{3}{2}-\tfrac{2}{3}=\tfrac{5}{6}$ and $M_{\Tc_2}=\tfrac{35}{12}-\tfrac{13}{6}=\tfrac{3}{4}$. These values can be either verified by enumeration and linearity of expectation, or else using the recursions of the upcoming Lemma~\ref{lem1}.
\end{example}

We will be interested in finding pairs of field gadgets analogous to Example~\ref{e4f334}. Our interest in such pairs of field gadgets is justified by the following theorem.
\begin{theorem}\label{thm:vanilla}
Let $\Delta\geq 3$ be an integer and $(\beta,\gamma,\lambda)\in \Uc_{\Delta}^*$. Suppose that there exists a pair of field gadgets $\Tc_1, \Tc_2$ with maximum degree $\Delta$ such that $R_{\Tc_1}=R_{\Tc_2}$ but $M_{\Tc_1}\neq M_{\Tc_2}$. 

Then, there is constant $c=c(\Delta,\beta,\gamma,\lambda)>1$ such that, for graphs $G$ of maximum degree $\Delta$, there is no poly-time $c$-approximation algorithm for computing the average magnetization $\Mc_{\beta,\gamma,\lambda}(G)$, unless \emph{P=NP}.
\end{theorem}

When can we find pairs of trees as in Example~\ref{e4f334} with the same effective fields but different magnetization gaps? Even in the case of the hard-core model,  a pair of trees either have the same effective field for only finitely many values of $\lambda$ or else have the same magnetization gap for all $\lambda$. Thus we cannot hope for a universal pair of gadgets.  Actually, the set of $\lambda$ where the effective fields can be equal, over all pairs of trees, must be algebraic, hence measure zero.

Our main theorem for the construction of field gadgets is the following. The theorem roughly asserts that we can construct field gadgets with arbitrarily close effective fields but substantially  different magnetization gaps. For a rational $r=p/q$ with integers $p,q$ satisfying  $\gcd(p,q)=1$, we let $\bit(r)$ denote the total number of bits needed to represent $p,q$. 
\begin{theorem}\label{thm:const1}
Let $(\beta,\gamma,\lambda)$ be antiferromagnetic with $(\beta,\gamma,\lambda)\neq (\beta,\beta,1)$. There exist $\hat{R},\hat{M}, \Xi>0$ and an algorithm which on input a rational $r\in (0,1/2)$ outputs in time $poly(\bit(r))$ a pair of field gadgets $\mathcal{T}_1,\mathcal{T}_2$, each of maximum degree 3 and size $O(|\log r|)$, such that 
\[|R_{\Tc_1} -\hat{R}|, |R_{\Tc_2}- \hat{R}|\leq r,\mbox{ but } |M_{\Tc_1}-M_{\Tc_2}|>\hat{M}.\]
Moreover, the magnetization gaps  $M_{\Tc_1}, M_{\Tc_2}$ are bounded in absolute value by the constant $\Xi$.
\end{theorem}
The proof of Theorem~\ref{thm:const1} is given in Section~\ref{sec:gadget}. In fact, we can bootstrap Theorem~\ref{thm:const1} to obtain pairs of trees with the same effective fields but different magnetization gaps for a dense set of (algebraic numbers) $\lambda$, and this gives a constant factor inapproximability result for the magnetization by applying Theorem~\ref{thm:vanilla}.
\begin{theorem}\label{thm:const2}
Let $(\beta,\gamma)$ be antiferromagnetic. There exists a set $S$ of algebraic numbers $\lambda$, dense in the interval $(0,\infty)$, such that for each $\lambda\in S$ the following holds. There is a pair of field gadgets $\mathcal{T}_1,\mathcal{T}_2$, each of maximum degree 3, such that $R_{\Tc_1}= R_{\Tc_2}$ but $M_1\neq M_2$.
\end{theorem}

In the following, we sketch the proof of Theorem~\ref{thm:vanilla} and in Section~\ref{sec:gadget} we will give a detailed outline of the proof of Theorems~\ref{thm:const1} and~\ref{thm:const2}.

\subsection{The reduction: using field gadgets to obtain inapproximability}

Fix $\Delta\geq 3$ and $\beta,\gamma,\lambda\in \Uc_{\Delta}$. Our results will be based on showing that approximating $\MaxCut$ on 3-regular graphs $H$ reduces to  approximating the magnetization will be via a reduction from $\MaxCut$. The reduction uses two types of gadgets. 

The first type of gadget is a bipartite graph $G$ which is an almost $\Delta$-regular graph with $n$ vertices on each side and $\ell$ ``ports'' on each side of degree $\Delta-1$, which will be used to connect distinct copies of gadgets. These gadgets were used in previous inapproximability results for the partition function and analysing them was one of the key difficulties in those results. The main idea in these results is to replace each vertex of $H$ by a distinct copy of $G$ and make appropriate connections between ports to encode the edges of $H$; then due to the antiferromagnetic interaction the resulting graph settles in configurations which correspond to Max-Cut configurations of $H$ (the Max-Cut assignment can roughly be read off by looking at the ``phases'' of the gadgets $G$, i.e., which side has the most occupied vertices,  see Section~\ref{sec:phase-gadget} for details).

However, for our inapproximability results we would need to analyze the magnetization of such gadgets by taking into account the effect of conditioning the spins of the ports; this task seems even more challenging than in the previous settings since  it seems to require very refined estimates. 

The field gadgets, which is the second type of gadgets, give a new reduction technique to bypass the need to perform this delicate, and likely difficult, task. In particular, by using a pair of field gadgets with the same effective fields and different magnetization gaps we will be able to observe the value of $\MaxCut(H)$ by taking the difference of the magnetizations when we append our field gadgets appropriately; crucially, the fact that the effective fields are the same implies that the underlying distribution does not change but that the difference of the magnetization gap of the two gadgets manifests itself in the magnetization.

\subsection{The bipartite ``phase-gadget'' of Sly and Sun}\label{sec:phase-gadget}
Fix an integer $\Delta\geq 3$. Let $\Gc^{\ell}_{n}$ be the set of $(2n)$-vertex bipartite graphs whose two sides are labelled with $\pl,\mi$ and are obtained  from a $\Delta$-regular $(2n)$-vertex bipartite graph by deleting a matching of size $\ell$. For a graph $G\in\Gc^{\ell}_{n}$,  we denote its bipartition by $(U^\pl,U^\mi)$ where $U^\pl,U^\mi$ are vertex sets with $|U^\pl|=|U^\mi|=n$, and we denote by $W^\pl,W^\mi$ the sets of vertices with degree $\Delta-1$ on each side of the bipartition, i.e., the vertices incident to the edges of the matching, so that $|W^\pl|=|W^\mi|=\ell$. We will refer to set $W=W^+\cup W^-$ as the \textit{ports} of $G$.

For $\sigma: U^\pl\cup U^\mi\rightarrow \{0,1\}$, we define the \emph{phase} $\Yc(\sigma)$ of the configuration $\sigma$ as the side which has the most occupied vertices under $\sigma$, i.e., 
\[\Yc(\sigma)=\begin{cases} \pl & \mbox{ if } |\sigma_{U^\pl}|\geq |\sigma_{U^\mi}|,\\ \mi& \mbox{otherwise}. \end{cases}\]
Sly and Sun \cite{Sly10} (see also \cite{GSV:colorings}) establish that, whenever  $\beta,\gamma,\lambda\in \Uc_{\Delta}$, for arbitrary $\epsilon>0$ and integer $\ell\geq 1$, there exist $n$ and $G\in \Gc^{\ell}_{n}$ such that, in the Gibbs distribution of the gadget graph $G$, each phase appears with probability close to 1/2$\pm \epsilon$, and that the spins of the ports of $G$, conditional on the phase, are roughly independent, with occupation probabilities that are asymmetric between the two sides. 

The detailed properties of the gadget can be found in Lemma~\ref{lem:SlySun} of Section~\ref{sec:SlySun}, though the exact details will not be important at this stage.

\subsection{The reduction}
Let $H$ be a 3-regular instance of the \textsc{Max-Cut} problem. 
Let $k$ be an arbitrary positive integer, and let $n$ and $G\in \Gc^{3k}_n$ be a bipartite gadget satisfying Lemma~\ref{lem:SlySun}. Let also $\mathcal{T}$ be a field gadget with effective field $R$ and magnetization gap $M$.

Let $\widehat{H}^{k}_G$ be the graph with $m$ disconnected components obtained by replacing each vertex of $H$ with a distinct copy of $G$; for $v\in V(H)$ we denote by $G_v$ the copy of $G$ in $\widehat{H}_G$ corresponding to the vertex $v$, and by $U^{\plm}_v, W^{\plm}_v$ the vertex sets and ports of $G_v$ in each side of the bipartition. Finally, we let $U_v=U^\pl_v\cup U^\mi_v$ and $W_v=W^\pl_v\cup W^\mi_v$.  

Let $H^{k}_{G,\Tc}$ be the graph obtained from $\widehat{H}^k_G$ as follows. For each edge $e=(u,v)\in E(H)$, pick $k$ distinct ports from $W^{\plm}_u, W^{\plm}_v$ and connect them using a path with three edges, for a total of $2k$ edge-disjoint paths between the gadgets $G_u$ and $G_v$: $k$ paths between $\pl/\pl$ sides and $k$ paths between $\mi/\mi$ sides.  Then, for each path $P$ append two distinct copies of the field gadget $\Tc$ by identifying the roots of the copies of $\Tc$, say $\rho_1,\rho_2$, with the internal vertices of $P$, say $t_1,t_2$. It will be helpful to consider also the graph $H^{k}_{G}$ which is the graph $H^{k}_{G,\Tc}$ when $\Tc$ is the single-vertex graph, i.e., $H^{k}_{G}$ is the graph after attaching the edge-disjoint paths but without the field gadgets.

For a configuration $\sigma: V(H^{k}_{G,\Tc})\rightarrow \{0,1\}$, we let $\widehat{\Yc}(\sigma): V(H)\rightarrow \{\pl,\mi\}$ be the phases over the gadgets $G_v$ with $v\in V(H)$, i.e., $\widehat{\Yc}(\sigma)=\big\{\Yc(\sigma_{U_v})\big\}_{v\in V(H)}$.  We define $\widehat{\Yc}(\cdot)$ similarly for configurations on $\widehat{H}^k_G$, $H^k_G$. For a phase assignment $Y:V(H)\rightarrow \{\pl,\mi\}$, we let $\mathrm{Cut}_H(Y)=\{(u,v)\in E(H)\mid Y_u\neq Y_v\}$
be the set of edges that are cut in $H$ by viewing the phase assignment $Y$ as a bipartition of $V(H)$ in the natural way. For $\mu=\mu_{H^{k}_{G,\Tc};\beta,\gamma,\lambda}$, let
\[\AvgCut_{\mu}(H):=\sum_{Y:V(H)\rightarrow \{\pl,\mi\}}\mu\big(\widehat{\Yc}(\sigma)=Y\big)|\Cut_H(Y)|\]
be the size of the average cut in $H$ when phase assignments are weighted by $\mu$. 

The following lemma associates the magnetization on the graph $H^{k}_{G,\Tc}$ with the quantity $\AvgCut_{\mu}(H)$ and, in turn, with $\MaxCut(H)$. The main idea is that the paths of length 3 between the ports cause antiferromagnetic interaction and the spin system on $H^{k}_{G,\Tc}$ prefers phase configurations $Y$ that have large $|\Cut_H(Y)|$ value.  By choosing $k$ large enough we can further ensure that $\AvgCut_{\mu}(H)$ arbitrarily close to $\MaxCut(H)$.

\begin{lemma}\label{lem:maingadget}
Let $\Delta\geq 3$ and $(\beta,\gamma,\lambda)\in \Uc_\Delta$. There are rational functions $A(R), B(R),C(R)$ defined for all $R\geq 0$, satisfying 
\begin{equation}\label{eq:vv6gvgg}
\mbox{(i) $A(0)=1$, (ii) $A(R)>1$ for $R>0$, and (iii)  $B(R)-C(R)=\lambda R\cdot (\log A(R))'$ for $R>0$}
\end{equation}
so that the following holds for any field gadget $\Tc$ with effective field $R>0$ and magnetization gap $M$.  Let $A,B,C$ be the values of the functions at $\lambda R$ and define $\Ac':=\Eb_{\tau\sim \mu_{\Tc;\beta,\gamma,\lambda}}[\,|\tau|\mid \tau(\rho)=0]$, where $\rho$ is the root of $\Tc$.

Let $\epsilon\in (0,1/10)$, $k\geq 10/\log A$ be an integer, and $n$ and $G\in \Gc^{3k}_n$ satisfy Lemma~\ref{lem:SlySun} with $\ell=3k$ and the given $\epsilon$. Then, for every 3-regular graph $H$ with sufficiently large $|V(H)|$,  for $\mu=\mu_{H^{k}_{G,\Tc};\beta,\gamma,\lambda}$, we have that the average magnetization of the graph $H^{k}_{G,\Tc}$ satisfies 
\begin{align*}
\Mc_{\beta,\gamma,\lambda}(H^{k}_{G,\Tc})&=4k\Ac'|E(H)|+\Eb_{\sigma\sim \mu}\big[\big|\sigma_{V(\widehat{H}^{k}_{G})}\big|\big]+(1\pm 8\epsilon)kMQ
\end{align*}
with $Q:=(B-C)\AvgCut_{\mu}(H)+C |E(H)|$, while   $\AvgCut_{\mu}(H)$ satisfies
\[1/K\leq \frac{\AvgCut_{\mu}(H)}{\MaxCut(H)}\leq 1 \mbox{\ \  for\ \  } K:=1+\frac{6}{k\log A}.\]
\end{lemma}
We defer the detailed proof of Lemma~\ref{lem:maingadget} to Section~\ref{sec:maingadget}.

Note that the expression for the magnetization of $H^{k}_{G,\Tc}$ accounts for the contributions of the phase gadgets $G$ only implicitly. This is by design; when we append our pair of field gadgets $\Tc_1,\Tc_2$ which have the same effective fields, vertices in $V(\widehat{H}^{k}_{G})$ will have the same marginal distribution in both $\mu_{H^{k}_{G,\Tc_1}}$ and $\mu_{H^{k}_{G,\Tc_2}}$; on the other hand the magnetization gaps of $\Tc_1,\Tc_2$ are different. Therefore, when we take the difference of the magnetizations in $H^{k}_{G,\Tc_1}$ and $H^{k}_{G,\Tc_2}$, the contributions of the phase gadgets $G$ cancel and, in the left-over quantity, the major contribution comes from the value of $\AvgCut_{\mu}(H)$.

We are almost ready to prove Theorem~\ref{thm:vanilla}. Recall, the assumption therein is that we have a pair of field gadgets with the same effective fields and different magnetization gaps. We will need to bootstrap this slightly to obtain additional  pairs of field gadgets, as stated in the following lemma.
\begin{lemma}\label{lem:vanilla}
Let $\Delta\geq 3$ be an integer and $(\beta,\gamma,\lambda)$ be antiferromagnetic with $(\beta,\gamma,\lambda)\neq (\beta,\beta,1)$. Suppose that there exists a pair of field gadgets $\Tc_1, \Tc_2$ with maximum degree $\Delta$ such that $R_{\Tc_1}=R_{\Tc_2}$ but $M_{\Tc_1}\neq M_{\Tc_2}$. Then, we can construct for $j=0,1,2,\hdots$ an infinite sequence of pairs of field gadgets $\Tc_{1,j},\Tc_{2,j}$ of maximum degree $\Delta$ with effective fields $R_{1,j}, R_{2,j}$ and magnetization gaps $M_{1,j},M_{2,j}$ such that $R_{1,j}=R_{2,j}$ but $M_{1,j}\neq M_{2,j}$; moreover the values of $R_{1,j}$, and hence of $R_{2,j}$ as well, are pairwise distinct. 
\end{lemma}
The proof of Lemma~\ref{lem:vanilla} is given in Section~\ref{sec:vanilla}. With Lemmas~\ref{lem:vanilla} and~\ref{lem:SlySun} at hand, we can now give the proof of Theorem~\ref{thm:vanilla}.

\begin{proof}[Proof of Theorem~\ref{thm:vanilla}]
For $i\in\{1,2\}$, let $R_i:=R_{\Tc_i}$ be the effective field of $\Tc_i$, $M_i:=M_{\Tc_i}$ be the magnetization gap of $\Tc_i$ and $\Ac_i':=\Eb_{\tau\sim \mu_{\Tc_i;\beta,\gamma,\lambda}}[\,|\tau|\mid \tau(\rho_i)=0]$. Note that $R_1=R_2$ but $M_1\neq M_2$, and $\Ac_1',\Ac_2'$ are absolute constants. 

 Let $A(R),B(R), C(R)$ be the functions in Lemma~\ref{lem:maingadget} and set $D(R)=B(R)-C(R)$. For convenience, we let $A,B,C,D$ denote the common values of $A(R_i),B(R_i),C(R_i),D(R_i)$, respectively, for $i\in \{1,2\}$.  We will need to ensure in our argument that $D\neq 0$. The first observation is that \eqref{eq:vv6gvgg} implies that for all but finitely many values of $R$ we have that $D(R)\neq 0$.\footnote{\label{ft:4f343}If $B(R)$ and $C(R)$ agree on infinitely many values of $R$, then since both of them are rational functions of $R$, it must be the case that $B(R)=C(R)$ for all $R\geq 0$. Since for all $R\geq 0$, we have that  $\lambda R\cdot (\log A(R))'=B(R)-C(R)$, this would give that $A(R)$ is a constant function; contradiction, since \eqref{eq:vv6gvgg} asserts that $A(0)=1$ and $A(R)>1$ for all $R>0$.} The second observation is that, from Lemma~\ref{lem:vanilla},  using $\Tc_1,\Tc_2$ we can construct for $j=0,1,2,\hdots$ an infinite sequence of pairs of field gadgets $\Tc_{1,j},\Tc_{2,j}$ of maximum degree $\Delta$ with effective fields $R_{1,j}, R_{2,j}$ and magnetization gaps $M_{1,j},M_{2,j}$ such that $R_{1,j}=R_{2,j}$ but $M_{1,j}\neq M_{2,j}$; moreover the values of $R_{1,j}$ are pairwise distinct. From these two observations, it follows that we can assume without loss of generality that $D\neq 0$.

Suppose for the sake of contradiction that, for arbitrarily small $\kappa>0$, there is a polynomial-time algorithm that, on input a graph $G$ of maximum degree $\Delta$, produces a $(1+\kappa)$-approximation of the magnetization $\Mc_{G;\beta,\gamma,\lambda}$. We will show that we can approximate $\MaxCut$ on 3-regular graphs within a constant factor arbitrarily close to 1, contradicting the inapproximability result of \cite{maxcut}.

For $i\in \{1,2\}$, consider the graph $H^{k}_{G,\Tc_i}$ for some large integer $k>1+10/\log A$ and small $\epsilon>0$ to be specified later. Let $\mu_i$ denote the Gibbs distribution on $H^{k}_{G,\Tc_i}$ with parameters $\beta,\gamma,\lambda$. By Lemma~\ref{lem:maingadget}, the average magnetization of the graph $H^{k}_{G,\Tc_i}$ satisfies 
\begin{equation}\label{eq:magsia}
\begin{aligned}
\Mc_{\beta,\gamma,\lambda}(H^{k}_{G,\Tc_i})&=4k\Ac_i'|E(H)|+\Eb_{\sigma\sim \mu_i}\Big[\big|\sigma_{V(\widehat{H}^{k}_{G})}\big|\Big]+(1\pm 8\epsilon)kM_{i} Q_i,
\end{aligned}
\end{equation}
where $Q_i=D\, \AvgCut_{\mu_i}(H)+C |E(H)|$ and $\AvgCut_{\mu_i}(H)$ satisfies
\begin{equation}\label{eq:avgcutia}
1/K\leq \displaystyle \frac{\AvgCut_{\mu_i}(H)}{\MaxCut(H)}\leq 1\mbox{ for }K:=1+\displaystyle\frac{6}{k\log A}.
\end{equation}
Note that since $R_1=R_2$ we have that (see the upcoming \eqref{eq:3dcevtvt5hyhuhu} for a more detailed explanation)
\begin{equation}\label{eq:t5g4rv55}
\Eb_{\sigma\sim \mu_1}\Big[\big|\sigma_{V(\widehat{H}^{k}_{G})}\big|\Big]=\Eb_{\sigma\sim \mu_2}\Big[\big|\sigma_{V(\widehat{H}^{k}_{G})}\big|\Big] \quad \mbox{ and }\quad \AvgCut_{\mu_1}(H)=\AvgCut_{\mu_2}(H).
\end{equation}
Let $\AvgCut_{\mu}(H)$ denote the common value of $\AvgCut_{\mu_1}(H),\AvgCut_{\mu_2}(H)$ and $Q$ be the common value of $Q_1,Q_2$.

Let $\Dc:=\Mc_{\beta,\gamma,\lambda}(H^{k}_{G,\Tc_1})- \Mc_{\beta,\gamma,\lambda}(H^{k}_{G,\Tc_2})$. From \eqref{eq:magsia} and \eqref{eq:t5g4rv55}
\[\Dc=4k(\Ac_1'-\Ac_2')|E(H)|+(1\pm 8\epsilon)k (M_1-M_2)Q.\] 

Let $L:=|V(\Tc_1)|+|V(\Tc_2)|$. Using the purported algorithm for the magnetization on $H^{k}_{G,\Tc_i}$, and since the latter graph has at most $|V(H)|\, |V(G)|+4k|E(H)|\, |V(\Tc_i)|\leq k(|V(G)|+4 L) |E(H)|=:X$ vertices, we can compute an estimate of $\Mc_{\beta,\gamma,\lambda}(H^{k}_{G,\Tc_i})$ that is off by at most (an additive) $\kappa X$. By subtracting these estimates for $i\in \{1,2\}$, we obtain an estimate $\widehat{\Dc}$ of $\Dc$ satisfying
\[|\widehat{\Dc}-\Dc|\leq 2k(|V(G)|+4 L) |E(H)|\kappa\]
for some small $\kappa$.  Then, we have that $\widehat{\textsc{MC}}=\frac{\widehat{\Dc}-4k(\Ac_1'-\Ac_2')|E(H)|}{k(M_1-M_2)D}-\frac{C}{D}|E(H)|$
satisfies 
\begin{equation}\label{eq:4g566g6}
|\widehat{\textsc{MC}}-\AvgCut_{\mu}(H)|\leq \Big(\tfrac{2(|V(G)|+4L)}{|M_1-M_2|}\kappa+8\big(1+\tfrac{|C|}{|D|}\big)\epsilon\Big)|E(H)|, 
\end{equation}
where $\AvgCut_{\mu}(H)$ denotes the common value of $\AvgCut_{\mu_1}(H),\AvgCut_{\mu_2}(H)$ (see \eqref{eq:t5g4rv55}). By choosing $k$ sufficiently large, we have from \eqref{eq:avgcutia} that $\AvgCut_{\mu}(H)$ is within a factor  arbitrarily close to 1 from  $\MaxCut(H)$. We will also choose $\epsilon>0$ to be arbitrarily close to 0. This has the potential effect of increasing the size of $G$, but  by choosing $\kappa>0$ to be sufficiently small we can nevertheless ensure from \eqref{eq:4g566g6} that our approximation $\widehat{\textsc{MC}}$ is within a factor  arbitrarily close to 1 from  $\AvgCut_{\mu}(H)$, and hence from $\MaxCut(H)$ as well. This finishes the contradiction argument and completes the proof of Theorem~\ref{thm:vanilla}.
 \end{proof}

With Theorem~\ref{thm:vanilla}, we can easily conclude the inapproximability result for the average-size of an independent set. Theorems~\ref{thm:mainind3},~\ref{thm:mainIs3},~\ref{thm:maingen3} will follow analogously, once we give the field gadget constructions of Section~\ref{sec:f434}. The proofs of these Theorems are given in Section~\ref{sec:f4343c}.
\begin{proof}[Proof of Theorem~\ref{thm:mainind1}]
By Example~\ref{e4f334}, we have two trees $\Tc_1$ and $\Tc_2$ with the same effective fields and different magnetization gaps. Also, independent sets correspond to $\beta=1,\gamma=0,\lambda=1$ and it is well-known (see, e.g., \cite{Sly10}) that  $(\beta,\gamma,\lambda)\in \Uc_{\Delta}^*$ iff $\Delta\geq 6$. Therefore, the desired inapproximability result follows by applying Theorem~\ref{thm:vanilla}.
\end{proof}

The proof of Theorems~\ref{thm:mainind2},~\ref{thm:mainIs2}, and~\ref{thm:maingen2} build on similar reduction ideas, but the details are slightly more technical because of the different type of gadgets that we use (cf. Theorem~\ref{thm:const1}). The details can be found in Section~\ref{sec:maingen2}.

\section{Field Gadget construction}\label{sec:gadget}

In this section, we give the proof of Theorem~\ref{thm:const1} by first outlining the key ideas of our field-gadget constructions. We start with the following lemma which describes the effective field and magnetization gap of field gadgets built out of smaller field gadgets. The proof is given in Section~\ref{sec:prooflem1}.
\begin{lemma}\label{lem1}
Suppose we have field gadgets $\mathcal{T}_1,\dots,\mathcal{T}_k$ with roots $\rho_1,\hdots,\rho_k$, effective fields $R_1,\dots,R_k$ and
magnetization gaps $M_1,\dots,M_k$, respectively. Let $\mathcal{T}$ be the tree with root $\rho$, edge
$\{\rho,u\}$ and the roots of $\mathcal{T}_1,\dots,\mathcal{T}_k$ identified with $u$. Let $R$ and $M$ be the
effective field and magnetization gap for $\mathcal{T}$. Then
\begin{equation*}
R = \frac{1+\gamma\lambda\prod^{k}_{i=1} R_i}{\beta+\lambda \prod^{k}_{i=1}R_i}, \qquad M  = 1 - \omega(R) \bigg( 1 + \sum_{i=1}^k (M_i-1)\bigg),
\end{equation*}
where
$$
\omega(R) := \frac{1+\beta \gamma-\beta R-\gamma/R}{1-\beta \gamma}.
$$
For antiferromagnetic $(\beta,\gamma)$ and any $\lambda>0$, it holds that $R\in (\gamma, 1/\beta)$ and
$0<\omega(R)<1$.
\end{lemma}

The first step in the construction is to create a family of field gadgets with sufficiently dense effective fields in
an interval and magnetization gaps that are in a small interval. To define the interval, consider $x^*$ and $\omega^*$ defined as follows
\begin{equation}\label{ode}
x^* = \frac{1+\gamma\lambda (x^*)^2}{\beta+\lambda (x^*) ^2},\quad \omega^* = \omega(x^*) = \frac{1+\beta \gamma-\beta x^*-\gamma/x^*}{1-\beta \gamma},
\end{equation}
and note  that $\omega^*\in(0,1)$ from Lemma~\ref{lem1}.

\begin{lemma}\label{led}
Fix $\lambda>0$. For any $\delta>0$ and any sufficiently small $\tau_1>0$ we can find $\tau\in (0,\tau_1)$ and construct a family of field gadgets $\mathcal{L}=\{{\cal T}_1,\dots,{\cal T}_k\}$ whose effective fields are in the interval $[x^* - \tau,x^*+\tau]$ and for any
$x\in [x^* - \tau,x^*+\tau]$ there exist a field gadget in the family whose effective field is in the
interval $[x-\tau\delta,x+\tau\delta]$.
\end{lemma}
The proof of Lemma~\ref{led}  builds on techniques from~\cite{complex} and is given in Section~\ref{sec:led}. 

The second step of the  construction is to use the field gadgets in ${\cal L}$ constructed in Lemma~\ref{led} in an iterative way as follows. At time $t=0$, we will start with some field gadget $\Tc_0$. Suppose that at some stage we have constructed a rooted field gadget $\Tc$ with effective field $R$ and magnetization gap $M$. To construct a new rooted field gadget, we merge $\Tc$ with one field gadget from ${\cal L}$ using the operation from Lemma~\ref{lem1}. We are going to analyze what pairs
(effective field, magnetization gap) we can achieve with this procedure. Let $R_i$ and $M_i$ be the effective
field and magnetization gap for the field gadget $\Tc_i$ in our collection. By Lemma~\ref{lem1}, merging the field gadget $\Tc$ with
$\Tc_i$ yields a rooted field gadget with effective field $\phi_i(R)$ and magnetization gap $\psi_i(R,M)$ where the pair of maps $(\phi_i,\psi_i)$ are given by
\begin{equation}\label{eq:1qazw}
\phi_i(R) = \frac{1+\gamma \lambda R R_i}{\beta +\lambda R R_i}\quad\mbox{and}\quad
\psi_i(R,M)= 1 - \omega(\phi_i(R))( M + M_i - 1).
\end{equation}
For a small constant $\tau>0$ to be specified later, let
\begin{equation}\label{idef2}
I' = \Big[x^* - 2 \tau \frac{|\omega^*|}{1-|\omega^*|}, x^* + 2 \tau \frac{|\omega^*|}{1-|\omega^*|}\Big]
\end{equation}
The choice of the interval $I'$ is such that the maps $\phi_i$ are uniformly contracting and map the interval $I'$ to itself. Namely, we show the following in Section~\ref{sec:contr}.
\begin{lemma}\label{contr}
There exists $0<C_{\mathrm{min}} < C_{\mathrm{max}}<1$ and $\tau_0>0$ such that for all $\tau\in(0,\tau_0)$
for all $R\in I'$ and all $R_i\in [x^* - \tau, x^* + \tau]$ it holds that
\[C_{\mathrm{min}}\leq |\phi'_i(R)|\leq C_{\mathrm{max}}, \quad \omega(R)\leq C_{\mathrm{max}}, \quad \phi_i(R)\in I'\]
\end{lemma}

The second step of our construction will actually take place in the following smaller sub-interval of $I'$:
\begin{equation}\label{idef}
I = [x^* - |\omega|\tau/2, x^* + |\omega|\tau/2] \subseteq I',
\end{equation}
The choice of $I$ is to ensure the following ``well-covered'' property
for the maps $\phi_i$ on $I$ which is obtained using the  density of the effective fields from Lemma~\ref{led}, the proof is given in Section~\ref{sec:contr}.
\begin{lemma}\label{lc}
Suppose $\delta<|\omega|/100$. For every two points $x_1,x_2\in I$ such that $|x_1-x_2|\leq |\omega|\tau \delta/2$
there exists $i$ such that $x_1,x_2\in\phi_i(I)$.
\end{lemma}

Lemmas~\ref{contr} and~\ref{lc} ensure that for any $x\in I$ we can construct a sequence of field gadgets whose effective field approaches $x$. Consider the following process \verb!Build-gadget!$(x,t)$ where $x\in I$ and $t\geq 0$ is an integer. If $t=0$
we return the degenerate tree. If $t\geq 1$ then we let $\phi_i$ be any map such that
$x\in\phi_i(I)$. Let $y = \phi_i^{-1}(x)$ and  $\Tc'=$\verb!Build-gadget!$(y,t-1)$. Return the tree $\Tc$ obtained by merging $\Tc'$ and $\Tc_i$ using the operation of Lemma~\ref{lem1}. 

The point behind the process \verb!Build-gadget!$(x,t)$ is that it allows us to construct, for arbitrary $x\in I$, a field gadget whose effective field is arbitrary close to $x$, with error that decays exponentially fast with $t$ (using the contraction properties of the $\phi_i$'s). This is detailed in the following lemma.
\begin{lemma}\label{yyyqw}
There exists $C>0$ such that for any $x\in I$ and any $t\geq 0$
the effective field $R$ of the field gadget returned by \verb!Build-gadget!$(x,t)$
(for any choice of the $\phi_i$'s) satisfies $|R-x|\leq C C_{\mathrm{max}}^t$.
\end{lemma}

We will prove Lemma~\ref{yyyqw} in Section~\ref{sec:case2}. For the field gadgets we construct we will always maintain the effective field in the interval $I'$, cf. Lemma~\ref{contr}. Now we show that the magnetization gaps of the field gadgets constructed using our process also stay restricted to an interval $J$. Let
\begin{equation}\label{TDEF}
T = \frac{2+\max|M_i|}{1-C_{\mathrm{max}}} \mbox{ and let  $J$ be the interval $[-T,T]$.}
\end{equation}
\begin{lemma}\label{eq:TDEF}
Suppose $R\in I'$ and $M\in J$ then $\psi_i(R,M)\in J$.
\end{lemma}
\begin{proof}
We have $\psi_i(R,M)  = 1 - \omega(R)( M_i + M - 1)$. Hence  $|\psi_i(R,M)|\leq 1 + C_{\mathrm{max}} ( T + \max |M_i| + 1) \leq T$, where the last inequality follows from~\eqref{TDEF}.
\end{proof}

For any $x\in I$ we are going to construct a sequence of families of pairs of maps
${\cal F}_{x,0},{\cal F}_{x,1},\dots$ as follows. In each pair, the first map is from $\mathbb{R}$ to $\mathbb{R}$,
and the second map is from $\mathbb{R}^2$ to $\mathbb{R}$ (similarly to \eqref{eq:1qazw}). Let ${\cal F}_{x,0}$ contain the pair of
maps $(x \mapsto x,(x,y)\mapsto y)$. To construct ${\cal F}_{x,t+1}$ we take every $\phi_i$
such that that $x\in \phi_i(I)$ and every $(f,g)\in {\cal F}_{\phi_i^{-1}(x),t}$ and place into ${\cal F}_{x,t+1}$
the map  
\begin{equation}\label{eq:remark}
(R,M)\mapsto \left( \phi_i(f(R)), \psi_i(f(R),g(R,M)) \right).
\end{equation}
Every $(f,g)\in {\cal F}_{x,t}$ corresponds to a sequence of $\phi_i$'s with length $t$, which in turn corresponds to a rooted field gadget built using the procedure described just above \eqref{eq:1qazw} for $t$ steps. In this procedure,  we will now view the initial field gadget $\Tc_0$ as an ``input'' (to the procedure). If the input has effective
field $R$ and magnetization gap $M$ then the final field gadget will have effective
field $f(R)$ and magnetization gap $g(R,M)$. In particular, with the right choice of input (i.e., $f^{-1}(x)$), we will obtain a field gadget with effective field $x$. We will usually view $f(R)$ and $g(R,M)$ as the effective field and the magnetization gap induced by the \verb!Build-gadget!$(x,t)$ procedure when at the base step we use $\Tc_0$ instead of the degenerate tree (provided the choice of the $\phi_i$'s in \verb!Build-gadget! matches up with the sequence for the pair $(f,g)$). However, we will not be able to use the exact input needed to obtain the effective field $x$ but rather some approximation of it, and the pair $(f,g)$ will allow us to track the influence of the last $t$ steps when building a rooted field gadget $\Tc'$, which intuitively have larger influence both on the effective field and magnetization gap of $\Tc'$ (since the maps $\phi_i$ are contracting).   In particular, once  we have the value of the magnetization gap and effective field of the ``input'' field gadget, by applying $(f,g)$ we know precisely where the effective field and magnetization gap will end up after applying the sequence of $\phi_i$'s corresponding to $(f,g)$.

We are going to distinguish two possible cases for the families ${\cal F}_{x,t}$. In the first case
we will obtain the gadgets we need immediately using the \verb!Build-gadget! procedure. In the second case we will construct a continuous
function from the families. Then we will argue that the function cannot satisfy a functional
equation and this will yield the gadgets we require.

\begin{lemma}[Case I]\label{case1}
Suppose there exists $x\in I$ such that for some $t_1,t_2$ there exist
$(f_1,g_1)\in {\cal F}_{x,t_1}$ and $(f_2,g_2)\in {\cal F}_{x,t_2}$
such that $g_1(I'\times J)$ and $g_2(I'\times J)$ are disjoint. Let $\hat{M}$ be the distance of
$g_1(I'\times J)$ and $g_2(I'\times J)$. 

Then, there is an algorithm which, on input a rational $r\in (0,1/2)$, outputs in time  $poly(\bit(r))$
a pair of field gadgets ${\cal T}_1$ and ${\cal T}_2$, each of maximum degree 3 and size $O(|\log r|)$, such that 
\[|R_{\Tc_1} - x|, |R_{\Tc_2} - x|\leq r, \mbox{ but } |M_{\Tc_1}-M_{\Tc_2}|\geq\hat{M}.\]
\end{lemma}

\begin{proof}
Note that for a fixed $t$ the union $\bigcup_x {\cal F}_{x,t}$ contains finitely many
functions (because there are only finitely many choices for $\phi_i$ in each
step of the construction). If $t_1,t_2$ and $x$ assumed by the lemma exist
then we can find them by examining finitely many functions. For each $(f_1,g_1)$
and $(f_2,g_2)$ we check whether $f_1(I)\cap f_2(I)\neq\emptyset$ and
$g_1(I'\times J)\cap g_2 (I'\times J) = \emptyset$; if we find we such a pair
we take $x\in f_1(I)\cap f_2(I)$. Note the running time for this process
is a constant depending on $\beta,\gamma,\lambda$. (Note that $f(I)$ and
$g(I'\times J)$ is always an interval that we can find inductively.)

To construct ${\cal T}_i$, $i\in\{1,2\}$ we use the \verb!Build-gadget! procedure following
the choices of $(f_i,g_i)$ in the first $t_i$ steps, cf. the discussion below \eqref{eq:remark}. Using Lemma~\ref{yyyqw}, we run the 
procedure for $O(|\log r|)$ steps achieving $|R_i - x|\leq r$, $i\in\{1,2\}$. 
We have $M_i\in g_i(I'\times J)$, $i\in\{1,2\}$ and hence $|M_1-M_2|\geq\hat{M}$.
\end{proof}

\begin{lemma}[Case II]\label{case2}
Suppose that for every $x$, every $t_1,t_2$ and every two functions
$(f_1,g_1)\in {\cal F}_{x,t_1}$ and $(f_2,g_2)\in {\cal F}_{x,t_2}$
we have that $f_2(I'\times J)$ and $g_2(I'\times J)$ intersect. Then there exists a continuous
function $F:I\rightarrow J$ such that for every $x\in I$, every $\eps>0$ there exists $t_0$ such that for every $t\geq t_0$ and
every $(f,g)\in {\cal F}_{x,t}$ and every $R\in I'$ and every $M\in J$ we have
\begin{equation}\label{close}
|g(R,M) - F(x)|\leq\eps.
\end{equation}
\end{lemma}
The proof of Lemma~\ref{case2} is given in Section~\ref{sec:case2}. 

\begin{lemma}\label{viola}
Suppose CASE II happens, that is, the assumption of Lemma~\ref{case2} is satisfied; let  $F$ be the continuous function 
guaranteed by Lemma~\ref{case2}. Suppose that there exist $x_1,x_2\in I$ such that the following equation is violated.
\begin{equation}\label{fun2}
F\left(\frac{1+\gamma\lambda x_1 x_2}{\beta+\lambda x_1x_2}\right) = 1  - \frac{(1-\beta \gamma)\lambda x_1 x_2}{(1+\gamma\lambda x_1 x_2)(\beta+\lambda x_1 x_2)}(F(x_1) + F(x_2) - 1)
\end{equation}
Then, for any integer $k\geq 1$, we can find in constant time (where the constant depends on $\beta,\gamma,\lambda,k$) rational numbers $x_1,x_2\in I$ such that \eqref{fun2} is violated and, moreover rationals $\hat{R}_1,\hdots,\hat{R}_k\in I$ such that the following holds.

There is an algorithm which on input a rational $r\in (0,1/2)$  and any $i\in [k]$ outputs in time $poly(\bit(r))$
a pair of field gadgets ${\cal T}_1$ and ${\cal T}_2$, each of maximum degree 3 and size $O(|\log r|)$, such that $|R_{\Tc_1} - \hat{R}_i|\leq r$,
$|R_{\Tc_2} - \hat{R}_i|\leq r$ and $|M_{\Tc_1}-M_{\Tc_2}|\geq\hat{M}$. 
\end{lemma}

\begin{proof}[Proof of Lemma~\ref{viola}]
Let $k\geq 1$ be an arbitrary integer and fix $x_1,x_2\in I$ that violate~\eqref{fun2}. Let $\delta>0$ be the absolute value of the difference between the two sides. From the continuity of $F$, more precisely 
equation~\eqref{equ2}, there exists $\eps>0$ such that for any $y_1,y_2\in I$ 
with $|y_1-x_1|\leq\eps$ and $|y_2-x_2|\leq\eps$ we have that $y_1,y_2$ 
violate~\eqref{fun2} with difference at least $\delta/2$ between the two sides.

Let $\mathcal{C}$ be a finite set of pairs $(y_1,y_2)$ which form an $\tfrac{\eps}{20k}$-net for $I\times I$; by perturbing slightly the set of points in $\mathcal{C}$ we can obtain an $\tfrac{\eps}{10k}$-net for $I\times I$, say $\mathcal{C}'$,  such that for any two pairs $(y_1,y_2)$ and $(y_1',y_2')$ it holds that $y_1y_2\neq y_1'y_2'$. Now, to check whether a pair $(y_1,y_2)$ in $\mathcal{C}'$ violates~\eqref{close}, we run 
\verb!Build-gadget! for $y_1,y_2$ and $y_3:=\frac{1+\gamma\lambda y_1 y_2}{\beta+\lambda y_1y_2}$ 
with the value of $t$ given by Lemma~\ref{case2} to achieve bound $\delta/100$
on the right-hand side of~\eqref{close}. Then, we will find at least $k$ different pairs $(y_{1,j},y_{2,j})$, $j=1,\hdots, k$ that violate~\eqref{fun2}. For $j\in [k]$, let  $\hat{R}_j:=y_{3,j}=\frac{1+\gamma\lambda y_{1,j} y_{2,j}}{\beta+\lambda y_{1,j}y_{2,j}}$ and note that by the construction of $\mathcal{C}'$, the $\hat{R}_j$'s are pairwise distinct. 

Now, on input $j\in [k]$, we use the \verb!Build-gadget! procedure to construct $\hat{\cal T}_i$ for $y_{i,j}$, $i\in\{1,2,3\}$ 
for $t=O(|\log r |)$ steps, using  Lemma~\ref{yyyqw},. The tree ${\cal T}_1$ is obtained by merging
$\hat{\cal T}_1$ and $\hat{\cal T}_2$ and the tree ${\cal T}_2$ is $\hat{\cal T}_3$. 
\end{proof}

Now assume that equation~\eqref{fun2} is satisfied for all $x_1,x_2\in I$. We are
going to derive a contradiction, thereby showing that ~\eqref{fun2} must be violated.
We will do this in two steps. First we use a special case of equation~\eqref{fun2}
to obtain a functional equation that constrains the possible solutions of $F$.
Second we show that none of these solutions satisfies~\eqref{fun2}.

Suppose $x_1,x_2,x_3\in I$ are such that $x_1 x_2 = x_3 x^*$. Plugging into~\eqref{fun2}
we obtain that $F$ has to satisfy the following equation
\begin{equation}\label{fun}
F(x_1) + F(x_2) = F(x_3) + F(x^*).
\end{equation}

\begin{lemma}\label{lem:3deecer}
Suppose $F$ is a continuous function on $I$. Suppose that for $x_1,x_2,x_3\in I$ such that
$x_1 x_2 = x_3 x^*$ we have~\eqref{fun}. Then there exists $c$ such that for all $x\in I$
we have
\begin{equation}\label{caus}
F(x) = c\log(x/x^*) + F(x^*).
\end{equation}
\end{lemma}

\begin{proof}
We will use the following parametrization to turn~\eqref{fun} into
Cauchy's functional equation. Let $x_1 = x^*\exp(y_1)$, $x_2 =2x^*\exp(y_2)$,
and $x_3=x^*\exp(y_3)$. The condition $x_1 x_2 = x^* x_3$ is equivalent
to $y_1+y_2=y_3$. Let
\begin{equation}\label{gdef}
G(y) = F(x^*\exp y) - F(x^*).
\end{equation}
 Note that $G$
is defined on the interval $[\log(I_L/x^*),\log(I_R/x^*)]$ that contains
$0$ (since $I_L<x^*<I_R$). Equation~\eqref{fun}
becomes
\begin{equation}\label{fun3}
G(y_1)+G(y_2) = G(y_1+y_2).
\end{equation}
From continuity of $F$ we have continuity of $G$. Since~\eqref{fun3}
is Cauchy's functional equation on an interval containing zero the
only continuous solutions are
$G(y)=cy$ for some constant $c$. Plugging in~\eqref{gdef} we obtain~\eqref{caus}.
\end{proof}

Finally we show that~\eqref{fun2} has to be violated.

\begin{lemma}\label{final}
A solution of the form~\eqref{caus} cannot satisfy~\eqref{fun2}.
\end{lemma}

\begin{proof}
We will plug-in~\eqref{caus} into~\eqref{fun2} with $x_1=x^*$ and $x_2=y$. We obtain
\begin{equation}\label{zzz2}
c \log \Big(\frac{1+\gamma\lambda x^* y}{\beta+\lambda x^* y}\Big) - c\log x^* + F(x^*) =
1- \frac{(1-\beta \gamma)\lambda x^* y}{(1+\gamma\lambda x^* y)(\beta+\lambda x^* y)}( c\log(y) - c\log(x^*)- 1 +  2 F(x^*)).
\end{equation}
Differentiating (w.r.t. $y$) we obtain
\begin{align*}
- c \frac{\lambda(1-\beta \gamma) x^*}{(\beta +\lambda x^* y)(1+\gamma\lambda x^* y)}& =
\left(-\frac{(1-\beta \gamma)\lambda x^*(-\beta+\gamma \lambda^2 (x^*y)^2) }{(1+\gamma\lambda x^* y)^2(\beta+\lambda x^* y)^2} \right)( c\log(y) - c\log(x^*)- 1 +  2 F(x^*))\\&\hskip 2cm
-\frac{(1-\beta \gamma)\lambda x^* y}{(1+\gamma\lambda x^* y)(\beta+\lambda x^* y)}\left( c\frac{1}{y} \right),
\end{align*}
which simplifies to
\begin{equation}\label{zzz}
0 =\frac{(1-\beta \gamma)\lambda x^*(-\beta+\gamma \lambda^2 (x^*y)^2) }{(1+\gamma\lambda x^* y)^2(\beta+\lambda x^* y)^2}( c\log(y) - c\log(x^*)- 1 +  2 F(x^*))
\end{equation}
The only solution for~\eqref{zzz} has to have $c=0$ and $F(x^*)=1/2$. Plugging $c=0$ and $F(x^*)=1/2$ into~\eqref{zzz2} we obtain 1/2 = 1 + 0, a contradiction.
\end{proof}

Finally, we combine the various pieces from the previous subsection to prove Theorem~\ref{thm:const1}. In fact, we will prove a slight strengthening of Theorem~\ref{thm:const1}, given below, which will be used in our reduction, analogously to Lemma~\ref{lem:vanilla}.
\begin{theorem}\label{thm:occgadget}
Let $(\beta,\gamma,\lambda)$ be antiferromagnetic with $(\beta,\gamma,\lambda)\neq (\beta,\beta,1)$. For every integer $k\geq 1$, there exist constants $\hat{M},\Xi>0$ and $k$ distinct numbers $\hat{R}_1,\hdots, \hat{R}_k>0$ such  that the following holds. There is an algorithm, which, on input $i\in[k]$ and a rational $r\in(0,1/2)$, outputs in time  $poly(\bit(r))$ a pair of field gadgets  $\mathcal{T}_1,\mathcal{T}_2$, each of maximum degree $3$ and size $O(|\log r|)$, such that 
\[|R_{\mathcal{T}_1} - \hat{R}_i|,\, |R_{\mathcal{T}_2} - \hat{R}_i|\leq r,\mbox{ but } |M_{\mathcal{T}_1}-M_{\mathcal{T}_2}|\geq \hat{M}.\]
Moreover, the magnetization gaps  $M_{\Tc_1}, M_{\Tc_2}$ are bounded in absolute value by the constant $\Xi$.
\end{theorem}
\begin{proof}
There are two cases to consider: either Lemma~\ref{case1} or Lemma~\ref{case2}. In the latter case, by Lemmas~\ref{lem:3deecer} and~\ref{final}, we obtain that \eqref{fun2} is violated, and hence Lemma~\ref{viola} yields the desired algorithm. In the former case, we are also done, modulo that Lemma~\ref{case1} only guarantees the existence of a single $\hat{R}$, namely the value  $x$. Let $t_1,t_2$ be such that $(f_1,g_1)\in {\cal F}_{x,t_1}$ and $(f_2,g_2)\in {\cal F}_{x,t_2}$ satisfy $g_1(I'\times J)\cap g_2(I'\times J)=\emptyset$. The functions in $\mathcal{F}_{x,t}$ are continuous and defined on a closed interval, so we have that for sufficiently small $\epsilon>0$, for all $y$ such that $|y-x|\leq \epsilon$, we can find $(\tilde{f}_1,\tilde{g}_1)\in {\cal F}_{y,t_1}$ and $(\tilde{f}_2,\tilde{g}_2)\in {\cal F}_{y,t_2}$  which satisfy $\tilde{g}_1(I'\times J)\cap \tilde{g}_2(I'\times J)=\emptyset$. We pick $k$ such $y$'s for the values of the $\hat{R}_i$'s for $i=1,\hdots,k$ and produce the desired trees by running the  \verb!Build-gadget! procedure for $O(\log |r|)$ steps. Finally note that the magnetizations all lie in the interval $J$, see \eqref{TDEF} and Lemma~\ref{eq:TDEF}, which is bounded by absolute constants, finishing the proof. 
\end{proof}
Note that Theorem~\ref{thm:const1}  corresponds to the case $k=1$ in Theorem~\ref{thm:occgadget}.

\subsection{Proof of Theorem~\ref{thm:const2}}
To prove Theorem~\ref{thm:const2}, we will just need an absolute bound on the second derivative of the effective-field constructions of the previous subsection, when viewed as functions of the parameter $\lambda$. The following lemma gives this bound.
\begin{lemma}\label{lem:34f34vrr}
Let $R_1(\lambda),\dots,R_k(\lambda)$ be a family of functions from a closed interval $I$ to $(\gamma,1/\beta)$. Consider the following family ${\cal F}$ of functions from $I$ to $(\gamma,1/\beta)$. Let $Q_0(\lambda)$ be a function from $I$ to $(\gamma,1/\beta)$. For every sequence $a_0,\dots,a_{m-1}$ of numbers in $[k]$ for $i\in [m-1]$ let
$$
Q_{i+1}(\lambda) = \frac{1+\gamma\lambda R_{a_i}(\lambda) Q_{i}(\lambda)}{\beta + \lambda R_{a_i}(\lambda) Q_{i}(\lambda)}.
$$
and place function $Q_m$ into ${\cal F}$.

Suppose each $R_\ell(\lambda)$ and $Q_0(\lambda)$ is twice continuously differentiable on $I$. There exists bounded
intervals $I_1,I_2$ such that for any $\lambda\in I$ and any $Q(\lambda)\in{\cal F}$ we have $Q_i'(\lambda)\in I_1$
and $Q_i''(\lambda)\in I_2$.
\end{lemma}

\begin{proof}
Let $\ell=a_i$. We have
\begin{equation}\label{azzz}
Q_{i+1}'(\lambda) =
\frac{(\gamma\lambda {R'}_\ell(\lambda) Q_i(\lambda)+\gamma\lambda R_\ell(\lambda) Q_i'(\lambda)) }{\beta + \lambda R_i(\lambda) Q_i(\lambda)}
-
\frac{(1+\gamma\lambda R_\ell(\lambda) Q_i(\lambda)) (\lambda {R'}_\ell(\lambda) R (\lambda)+\lambda R_\ell(\lambda) Q_i'(\lambda)) }{(\beta + \lambda R_\ell(\lambda) Q_i(\lambda))^2}.
\end{equation}
We can re-write~\eqref{azzz} as follows
\begin{equation}\label{ind}
\frac{Q_{i+1}'(\lambda)}{Q_{i+1}(\lambda)} = \frac{\lambda R_\ell(\lambda)Q_i(\lambda) (\beta\gamma-1)}{(1+\lambda\gamma R_\ell(\lambda)Q_i(\lambda))(\beta+\lambda R_\ell(\lambda)Q_i(\lambda))} + \frac{p_1(R_\ell(\lambda),Q_i(\lambda),R'_\ell(\lambda))}{(\beta + \lambda R_\ell(\lambda) Q_i(\lambda))},
\end{equation}
where $p_1$ is a constant degree polynomial. Note that the coefficient of $Q_i'(\lambda)/Q_i(\lambda)$ in~\eqref{azzz} satisfies
\begin{equation}\label{azzz2}
\Big|\frac{\lambda R_\ell(\lambda)Q_i(\lambda) (\beta\gamma-1)}{(1+\lambda\gamma R_\ell(\lambda)Q_i(\lambda))(\beta+\lambda R_\ell(\lambda)Q_i(\lambda))}\Big| \leq C < 1,
\end{equation}
where $C=\frac{1-\sqrt{\beta\gamma}}{\sqrt{\beta\gamma}+1}$ if $\beta\gamma\neq 0$ (the case $\beta\gamma=0$ is
discussed below). To see the bound in~\eqref{azzz2} let $W=\lambda R_\ell(\lambda)Q_i(\lambda)$. The coefficient becomes $\frac{W(\beta\gamma-1)}{(1+\gamma W)(\beta + W)}$, which achieves maximum for $W=\sqrt{\beta/\gamma}$ with value $\frac{\sqrt{\beta\gamma}-1}{\sqrt{\beta\gamma}+1}$. In the
case $\beta=0$ the coefficient becomes $-1/(W\gamma+1)$, which since $W$ is uniformly bounded from below 
has absolute value uniformly bounded from $1$. In the case $\gamma=0$ the coefficient becomes $-W/(\beta+W)$
and since $W$ is uniformly bounded from above has absolute value uniformly bounded from $1$.

Each $R'_\ell(\lambda)$ is bounded (since it is
a continuous function on a closed interval. Hence there exists constant $M_1$ such that for
any $x\in (\gamma,1/\beta)$ and any $\lambda\in I$
$$
\left|
\frac{p_1(R_\ell(\lambda),x,R'_\ell(\lambda))}{(\beta + \lambda R_\ell(\lambda) x)^2}
\right|\leq M_1.
$$
From continuous differentiability of $Q_0$ we also have that there exists a constant $M_2$ such
that for any $\lambda\in I$
we have
$$
\Big|\frac{Q_0'(\lambda)}{Q_0(\lambda)}\Big|\leq M_2.
$$
Now using~\eqref{ind} by induction $|\frac{Q_{i}'(\lambda)}{Q_{i}(\lambda)}|\leq (M_1+M_2)/(1-C)$.

The argument for the second derivative is almost the same. We can write an expression
for $Q''_{i+1}(\lambda)/Q_{i+1}(\lambda)$ as follows
\begin{equation}\label{ind2}
\frac{Q_{i+1}''(\lambda)}{Q_{i+1}(\lambda)} = \frac{\lambda R_\ell(\lambda)Q_i(\lambda) (\beta\gamma-1)}{(1+\lambda\gamma R_\ell(\lambda)Q_i(\lambda))(\beta+\lambda R_\ell(\lambda)Q_i(\lambda))} \frac{Q_{i}''(\lambda)}{Q_i(\lambda)} + \frac{p_2(R_\ell(\lambda),Q'_i(\lambda),Q_i(\lambda),R'_\ell(\lambda),R''_\ell(\lambda))}{(\beta + \lambda R_\ell(\lambda) Q_i(\lambda))^3},
\end{equation}
where $p_2$ is a constant degree polynomial. Again all the parameters for $p_2$
are bounded (for $Q_i'(\lambda)$ we use the $(M_1+M_2)/(1-C) (1/\beta)$ bound) and hence
there exists a constant $M_3$ such that for any $x\in (\gamma,1/\beta)$ and any
$|y|\leq (M_1+M_2)/(1-C)$ we have
$$
\Big|\frac{p_2(R_\ell(\lambda),y,x,R'_\ell(\lambda),R''_\ell(\lambda))}{(\beta + \lambda R_\ell(\lambda) x)^3}\Big|\leq M_3.
$$
We also have that there exists a constant $M_4$ such
that for any $\lambda\in I$
we have
$$
\Big|\frac{Q_0''(\lambda)}{Q_0(\lambda)}\Big|\leq M_4.
$$
Note that the coefficient of $Q_{i}''(\lambda)/Q_{i}(\lambda)$ in~\eqref{ind2} satisfies~\eqref{azzz2}. Now using~\eqref{ind2} by induction $|Q_{i}''(\lambda)/Q_i(\lambda)|\leq (M_3+M_4)/(1-C)$.
\end{proof}

We now give the proof of Theorem~\ref{thm:const2}.

\begin{proof}[Proof of Theorem~\ref{thm:const2}]
By Theorem~\ref{thm:const1}, for antiferromagnetic $(\beta,\gamma)$ and $\lambda>0$ with $(\beta,\gamma,\lambda)\neq (\beta,\beta,1)$, there exist $\hat{R},\hat{M}>0$ such that  for any $r>0$, there exist  a pair of field gadgets $\mathcal{T}_1,\mathcal{T}_2$, each of maximum degree 3 and size $O(|\log r|)$, such that 
\[|R_{\Tc_1} -\hat{R}|, |R_{\Tc_2}- \hat{R}|\leq r,\mbox{ but } |M_{\Tc_1}-M_{\Tc_2}|>\hat{M}.\]
The effective fugacities of the gadgets $\Tc_1$, $\Tc_2$, viewed as functions of $\lambda$, are of the form given in Lemma~\ref{lem:34f34vrr} and hence their second derivatives are uniformly bounded. Note that the magnetization gap is the first derivative of the logarithm of the effective fugacity. So, the lower bound on the difference of $M_{\Tc_1},M_{\Tc_2}$ implies an absolute lower bound on the difference of the derivatives of $R_{\Tc_1},R_{\Tc_2}$. Since the second derivative is uniformly bounded (by Lemma~\ref{lem:34f34vrr}), we can find $\hat{\lambda}$ arbitrarily close to $\lambda$ such that  $R_{\Tc_1}(\hat{\lambda}),R_{\Tc_2}(\hat{\lambda})$ are equal. Since  $R_{\Tc_1},R_{\Tc_2}(\lambda)$ are rational functions of $\lambda$ (using that $\beta,\gamma$ are rational), we obtain that any such $\hat{\lambda}$ must in fact be algebraic, finishing the proof.
\end{proof}

\section{Field Gadgets: Remaining Proofs}

In this section, we give the remaining proofs of Section~\ref{sec:gadget}.
\subsection{Proof of Lemma~\ref{lem1}}\label{sec:prooflem1}
\begin{proof}[Proof of Lemma~\ref{lem1}]
For convenience, we will write $\mu_{\Tc}$ instead of $\mu_{\Tc;\beta,\gamma,\lambda}$. Let $\Zin_{\Tc,\rho}$ be the weight of all configurations on $\Tc$ rooted at $\rho$ where $\rho$ is occupied (has state $1$). Similarly let $\Zout_{\mathcal{T},r}$ be the weight of all configurations on $\Tc$ rooted at $\rho$ where $\rho$ is unoccupied (has state $1$). Define analogously $\Zin_{\Tc_i,\rho_i}$ for $i=1,\hdots,k$. Then, we have
\begin{align*}
\Zin_{\Tc,\rho}&=\lambda \bigg(\prod^{k}_{i=1} \Zout_{\Tc_i,\rho_i}+\gamma\lambda\prod^{k}_{i=1} \frac{\Zin_{\Tc_i,\rho_i}}{\lambda}\bigg)\\
\Zout_{\Tc,\rho}&=\beta\prod^{k}_{i=1} \Zout_{\Tc_i,\rho_i}+\lambda \prod^{k}_{i=1} \frac{\Zin_{\Tc_i,\rho_i}}{\lambda}
\end{align*}
We therefore have that
\[R=\frac{1}{\lambda}\frac{\Zin_{\Tc,\rho}}{\Zout_{\mathcal{T},\rho}}=\frac{\prod^{k}_{i=1} \Zout_{\mathcal{T}_i,r_i}+\gamma\lambda\prod^{k}_{i=1} \frac{\Zin_{\Tc_i,\rho_i}}{\lambda}}{\beta\prod^{k}_{i=1} \Zout_{\Tc_i,\rho_i}+\lambda \prod^{k}_{i=1} \frac{\Zin_{\Tc_i,\rho_i}}{\lambda}}=\frac{1+\gamma\lambda\prod^{k}_{i=1} R_i}{\beta+\lambda \prod^{k}_{i=1}R_i}.\]
For the magnetization gap, we view $\Zin_{\Tc,\rho}, \Zout_{\Tc,\rho},$ and $\Zin_{\Tc_i,\rho_i}, \Zout_{\Tc_i,\rho_i}$ as polynomials in $\lambda$, and therefore the $R_i$'s and $R$ as rational functions of $\lambda$. Then observe that
\[M_i=\lambda \frac{\partial \log \Zin_{\Tc_i,\rho_i}}{\partial \lambda}- \lambda \frac{\partial \log \Zout_{\Tc_i,\rho_i}}{\partial \lambda}=\lambda \frac{\partial \log (\lambda R_i)}{\partial \lambda}=1+\frac{\lambda}{R_i} \frac{\partial R_i}{\partial \lambda},\]
and similarly $M=1+\frac{\lambda}{R} \frac{\partial R}{\partial \lambda}$. It follows that
\begin{align*}
M&=1-\frac{(1-\beta \gamma)\lambda}{R(\beta+\lambda \prod^{k}_{i=1}R_i)^2}\frac{\partial (\lambda \prod^{k}_{i=1}R_i)}{\partial \lambda}=1-\frac{(1-\beta\gamma)\lambda}{(1+\gamma\lambda\prod^{k}_{i=1} R_i)(\beta+\lambda \prod^{k}_{i=1}R_i)}\frac{\partial (\lambda \prod^{k}_{i=1}R_i)}{\partial \lambda}\\
&=1-\frac{(1-\beta \gamma)\lambda \prod^{k}_{i=1}R_i}{(1+\gamma\lambda\prod^{k}_{i=1} R_i)(\beta+\lambda \prod^{k}_{i=1}R_i)}\bigg(1+\sum^{k}_{i=1}\frac{\lambda}{R_i}\frac{\partial R_i}{\partial \lambda}\bigg)=1 -\omega(R)  \bigg( 1 + \sum_{i=1}^k (M_i-1)\bigg).\qedhere
\end{align*}
\end{proof}

\subsection{Proof of Lemma~\ref{led}}\label{sec:led}
In this section, we give the proof of Lemma~\ref{led}. We first create two field gadgets ${\cal T}_1$ and ${\cal T}_2$ with effective fields $R_1,R_2$ where $R_1$ is close to $x^*$ and $R_2$ is even closer to $x^*$. We generalize a ``bounding-tree'' construction from \cite{glauber}, which was used in the case of the hard-core model.

To handle the special case $\lambda=\frac{1-\beta}{1-\gamma}$ for some $\beta,\gamma\in (0,1)$, cf. Definition~\ref{def:field} and the remark below, we will use the following field gadget whose effective field is different than $ 1$ and satisfies $R\in (\gamma,1/\beta)$.
\begin{lemma}\label{lem:triangle}
Let $(\beta,\gamma)$ be antiferromagnetic with $\beta,\gamma\in (0,1)$ and $\beta\neq \gamma$. Let $\lambda=\frac{1-\beta}{1-\gamma}$ and consider the field gadget $\Tc$ rooted at $\rho$, consisting of a 4-cycle together with the extra vertex $\rho$ that is attached to one of the vertices of the cycle. Then, the effective field $R$ of $\Tc$ is different than 1 and satisfies $R\in (\gamma,1/\beta)$.
\end{lemma} 
\begin{proof}
A direct calculation gives that $R=\frac{1+\gamma \lambda r}{\beta+\lambda r}$ where $r=\frac{\beta^2 +\lambda+ 2 \lambda \beta \gamma + 3\lambda^2 \gamma^2+\lambda^3 \gamma^4}{\beta^4 + 3\lambda \beta^2+\lambda^2+2 \lambda^2 \beta \gamma+ \lambda^3 \gamma^2}$. Since $\lambda=\frac{1-\beta}{1-\gamma}$, we have that $R\neq 1$ iff $r\neq 1$. Using the value of $\lambda$ again we obtain that $r=1$ is equivalent to  $(1-\beta)^3(\beta-\gamma)=0$, which cannot be the case when $\beta\neq \gamma$ and $\beta,\gamma\in(0,1)$. The fact that $R\in (\gamma,1/\beta)$ follows from $r\in (\gamma,1/\beta)$.
\end{proof}

\begin{lemma}\label{lem:rr44f}
For every $\eps_1,\eps_2>0$ there exist field gadgets $\mathcal{T}_1$ and $\mathcal{T}_2$ such that
\begin{equation}
|R_1 - x^*| < \eps_1 \quad\mbox{and}\quad |R_2 - x^*| < \eps_2 |R_1 - x^*|.
\end{equation}
\end{lemma}

\begin{proof}
We will inductively construct two sequences of field gadgets ${\cal T}^L_1,{\cal T}^L_2,\dots$ and
${\cal T}^U_1,{\cal T}^U_2,\dots$ with effective fields $L_1,L_2,\dots$
and $U_1,U_2,\dots$. Throughout the construction we will maintain $1/\beta < L_i\leq x^*\leq U_i < \gamma$
and $L_i\neq U_i$. We will use the following inequalities
\begin{equation}\label{eew}
x<x^*\implies \frac{1+\gamma\lambda x^2}{\beta+\lambda x^2}>x^*\quad\mbox{and}\quad
x>x^*\implies \frac{1+\gamma\lambda x^2}{\beta+\lambda x^2}<x^*;
\end{equation}
these follow from the fact that $\frac{1+\gamma\lambda x^2}{\beta+\lambda x^2}$ is decreasing
and that $x^*$ is the fixpoint.

If $1<x^*$ then we let ${\cal T}^L_1$ to be the degenerate tree (see Example~\ref{eee})
and ${\cal T}^U_1$ be the tree with one edge. Note that ${\cal T}^U_1$ is obtained
from the degenerate tree using the operation of Lemma~\ref{lem1} and hence $L_1 = 1$
and $U_1 = (1+\gamma\lambda)/(\beta+\lambda)$ which by~\eqref{eew} is greater than $x^*$.
In the case $x^*<1$ by the same argument we can take ${\cal T}^L_1$ the tree with one
edge and ${\cal T}^U_1$ the degenerate tree and we again have $L_1 < x^* < U_1$. In the case where 
$x^*=1$ we must have that $\lambda=\frac{1-\beta}{1-\gamma}$ for some $\beta,\gamma\in (0,1)$ with $\beta\neq \gamma$; then, we can use the field gadget $\Tc$ of Lemma~\ref{lem:triangle}: if $R<1$ then we let  ${\cal T}^L_1$ be $\Tc$ and we take ${\cal T}^U_1$ to be the single edge tree, and if $R>1$, vice versa.

Now assume $L_i\leq x^* \leq U_i$ and $L_i\neq U_i$. Using~\eqref{eew} we have that
\begin{equation}\label{iii}
\frac{1+\gamma\lambda U_i^2}{\beta+\lambda U_i^2} \leq  x^* \leq \frac{1+\gamma\lambda L_i^2}{\beta+\lambda L_i^2}.
\end{equation}
From the monotonicity of the map $x\mapsto (1+\gamma\lambda x)/(\beta+\lambda x)$ and~\eqref{iii}
we have that there is $k\in\{0,1\}$ such that
\begin{equation}\label{zzh}
L_{i+1}:=\frac{1+\gamma\lambda U_i^k L_i^{2-k}}{\beta+\lambda U_i^k L_i^{2-k}} \leq  x^* \leq \frac{1+\gamma\lambda U_i^{k+1}L_i^{1-k}}{\beta+\lambda U_i^{k+1} L_i^{1-k}}=:U_{i+1}.
\end{equation}
Equation~\eqref{zzh} corresponds to 1) creating ${\cal T}^L_{i+1}$ using $k$ copies of
${\cal T}^U_{i}$ and $2-k$ copies of ${\cal T}^L_{i}$ and 2) creating ${\cal T}^U_{i+1}$
using $k+1$ copies of ${\cal T}^U_{i}$ and $1-k$ copies of ${\cal T}^L_{i}$.

Note that $U_i\neq L_i$ implies $U_{i+1}\neq L_{i+1}$. We will now argue
that $\lim_{n\rightarrow\infty} L_n = \lim_{n\rightarrow\infty} U_n = x^*$. From~\eqref{zzh} we obtain
\begin{equation}\label{dac2}
\frac{U_i}{L_i} = \frac{(L_{i+1}\beta - 1)(\gamma-U_{i+1})}{(U_{i+1}\beta - 1)(\gamma-L_{i+1})}.
\end{equation}
We will now argue
\begin{equation}\label{dac}
\frac{(L_{i+1}\beta - 1)(\gamma-U_{i+1})}{(U_{i+1}\beta - 1)(\gamma-L_{i+1})} > \frac{U_{i+1}}{L_{i+1}}.
\end{equation}
We can rewrite~\eqref{dac} as follows
$$
(U_{i+1}-L_{i+1})( \gamma(1-\beta L_{i+1}) + \beta L (U_{i+1} - \gamma) ) > 0,
$$
which follows from $R\in (\gamma,1/\beta)$ (see Lemma~\ref{lem1}). Note that $U_i/L_i$
is a decreasing sequence, $U_i/L_i\geq 1$ and hence it has a limit. We are going to
argue that the limit is $1$. Suppose the limit is $\alpha\geq 1$. Note that the values
of $L_i$ are restricted to $(\gamma,1/\beta)$ and hence the set $\{L_i\,|\,i\geq 1\}$ has
an accumulation point $x$. At the accumulation point we have
\begin{equation}\label{eeeex}
\alpha = \frac{(x\beta-1)(\gamma-\alpha x)}{(\alpha x \beta -1)(\gamma-x)},
\end{equation}
since for every $\eps>0$ there exist $i$ such that $|L_{i+1}-x|\leq\eps$ and $|U_{i+1} - \alpha x|\leq \eps$
and $|U_{i}/L_{i}-\alpha|\leq\eps$ (which together with~\eqref{dac2} implies~\eqref{eeeex}).

Equation~\eqref{eeeex} simplifies to
\begin{equation}\label{eeeex2}
(1-\alpha)( \alpha\beta\gamma x -\alpha\beta x^2 +\beta\gamma x - \gamma ) =0.
\end{equation}
We will show that for $x\in (\gamma,1/\beta)$ we have
\begin{equation}\label{part}
\alpha\beta\gamma x -\alpha\beta x^2 +\beta\gamma x - \gamma < 0.
\end{equation}
For $x=\gamma$ the left-hand side of~\eqref{part} simplifies to $(\beta\gamma-1)\gamma<1$.
The derivative of the left-hand side of~\eqref{part} with respect to $x$ is
\begin{equation}\label{part2}
\beta (\alpha\gamma-2\alpha x + \gamma).
\end{equation}
We will show that for $x\in (\gamma,1/\beta)$ the expression~\eqref{part2} is negative;
because of linearity we only need to argue at the endpoints of the interval.
For $x=\gamma$ the value is $\gamma-\alpha\gamma<0$. For $x=1/\beta$ the
value is $(1+\alpha)\beta \gamma - 2\alpha<0$. This proves~\eqref{part2} and
hence~\eqref{eeeex2} implies $\alpha=1$, that is, the limit of $U_i/L_i$ goes
to $1$ as $i\rightarrow\infty$.

To construct ${\cal T}_1$ we iterate the process until $U_i/L_i\leq 1+\eps_1/x^*$;
this implies $|U_i-x^*|\leq\eps_1$ and $|L_i-x^*|\leq\eps_1$. Since $U_i\neq L_i$
at least one of them is different from $x^*$; this will be the effective
field of our field gadget ${\cal T}_1$. To construct ${\cal T}_2$ we iterate the
process until $U_j/L_j\leq 1+\eps_2|R_1-x^*|/x^*$ and take either the field gadget
for $L_j$ or $U_j$.
\end{proof}
\begin{proof}[Proof of Lemma~\ref{led}]
We will use the Taylor approximation of the map of Lemma~\ref{lem1} for $R_1,R_2$ close to $x^*$. More precisely, there exists $C>0$ and $\tau_0>0$ such that for any $\tau\in(0,\tau_0)$, $|y_1|\leq\tau$, and $|y_2|\leq\tau$ for $R_1=x^*+y_1$ and $R_2=x^*+y_2$ we have
\begin{equation}\label{oooz}
\Big|\frac{1+\gamma \lambda R_1 R_2}{\beta +\lambda R_1 R_2}
- R\Big|\leq C \tau^2,
\end{equation}
where $R=x^* - \omega^*(y_1 + y_2)$ (see equation~\eqref{bou1}).

For $i\in\{0,\dots,Q\}$ (where $Q=\lceil 2\log_2(1/\delta) + 2\log_2(1/\omega^*)\rceil$)
we will construct a sequence of field gadgets ${\cal T}_{i,0}, \dots, {\cal T}_{i,2^i}$ with effective
fields $R_{i,0},\dots,R_{i,2^i}$ that are close
to an arithmetic sequence. More precisely for each $i\in\{0,\dots,Q\}$
there will be a value $\alpha_i$ and $\xi_i>0$ such that
for all $j\in\{0,\dots,2^i\}$ we have
\begin{equation}\label{suho}
|R_{i,j} - (x^* + j\alpha_i)|\leq\xi_i |\alpha_i|.
\end{equation}
We will aim to have $\xi_i$ small; the exact bound will be established inductively. To simplify 
the notation we assume $C\geq 1$ (increasing $C$ preserves~\eqref{oooz}) and $\tau_0<1/100$
(decreasing $\tau_0$ preserves~\eqref{oooz}).

We start with ${\cal T}_{0,0}={\cal T}_2$ and ${\cal T}_{0,1}={\cal T}_1$ where ${\cal T}_1,{\cal T}_2$
are the field gadgets guaranteed by Lemma~\ref{lem:rr44f} with parameters
$\eps_2=\tau_0(\omega^*)^{2Q}/(C^2 9^Q)$, $\eps_1=\eps_2^2$. Note that
we satisfy~\eqref{suho} with $\alpha_0=R_{0,1}$ and $\xi_0 = \eps_2$.

Suppose we constructed ${\cal T}_{i,0}, \dots, {\cal T}_{i,2^i}$ and $\xi_i\leq 3^i\xi_0<1/100$. For $j\in\{0,\dots,2^{i+1}\}$ we
let ${\cal T}_{i+1,j}$ be obtained from ${\cal T}_{i,k}$ and ${\cal T}_{i,\ell}$ where $k+\ell=j$ (any choice 
of $k$ and $\ell$ is fine).
We let $\alpha_{i+1}=-\omega^*\alpha_i$. From~\eqref{oooz} we have
$$
\big| R_{i+1,j} - \big(x^* - \omega^*\big( (R_{i,k}-x^*) + (R_{i,\ell}-x^*) \big)\big) \big| \leq C (2^{i+1} |\alpha_i|) ^2,
$$
which combined with~\eqref{suho} for $R_{i,k}$ and $R_{i,\ell}$ implies
\begin{equation}\label{suho2}
|R_{i+1,j} - (x^* + j\alpha_{i+1})|\leq C 2^{2i+2} |\alpha_i| \xi_0 \xi_i + 2 \xi_i \omega^* |\alpha_i|
\leq 3 \xi_i |\alpha_{i+1}|,
\end{equation}
where in the first inequality we used $|\alpha_i|\leq |\alpha_0|\leq\eps_1 = (\xi_0)^2\leq \xi_0 \xi_i$.
Hence we have ${\cal T}_{i+1,0}, \dots, {\cal T}_{i+1,2^{i+1}}$ that satisfy~\eqref{suho} with $\xi_{i+1}\leq 3^{i+1}\xi_0<1/100$.

Our final family of field gadgets is ${\cal T}_{Q,i},i\in\{0,\dots,2^Q\}$ and ${\cal T}_{Q-1,i}, i\in\{0,\dots,2^{Q-1}\}$
(the reason we take field gadgets for two values of the first index is to cover both sides of $x^*$; note that
$\alpha_i$'s alternate sign). We take $\tau = 2^{Q-1}\omega^*|\alpha_{Q}|$. 
Let $x\in [x^*-\tau,x^*+\tau]$. If the sign of $x-x^*$ agrees with the sign of $\alpha_Q$ we
let $k = \lfloor (x-x^*)/\alpha_Q\rfloor$ (note $0\leq k\leq 2^{Q-1}$) and take the tree ${\cal T}_{Q,k}$.
We have 
\begin{equation}\label{uuuup}
| R_{Q,k} - (x^* + \alpha_Q k) | \leq (1/100) |\alpha_Q|
\end{equation}
and hence 
\begin{equation}\label{uuuuq}
| R_{Q,k} - x | \leq 2 |\alpha_Q| \leq 2^{-Q+2} \tau \leq\delta\tau.
\end{equation}
If the sign of $x-x^*$ agrees with the sign of $\alpha_{Q-1}$ we
let $k = \lfloor (x-x^*)/\alpha_{Q-1}\rfloor$ (note $0\leq k\leq 2^{Q-1}$) and take 
the tree ${\cal T}_{Q-1,k}$. The argument is analogous to~\eqref{uuuup} and~\eqref{uuuuq}.
\end{proof}

\subsection{Proof of Lemmas~\ref{contr} and~\ref{lc}}\label{sec:contr}

\begin{proof}[Proof of Lemma~\ref{contr}]
We will consider the functions $\phi_i$ for $R$ and $R_i$ close to $x^*$.  The gradient of $\phi_i(R)$ (viewed as a function of $R$ and $R_i$) is
\begin{equation}\label{deri}
\left(\frac{\lambda R_i (\beta\gamma-1)}{(\beta +\lambda R R_i)^2}, \frac{\lambda R (\beta\gamma-1)}{(\beta +\lambda R R_i)^2}\right),
\end{equation}
so we let $R_i=R=x^*$ in~\eqref{deri} then we obtain that the gradient of $\phi_i$ at $(x^*,x^*)$ is
$(\omega^*,\omega^*)$, since
\begin{equation*}
\frac{\lambda x^* (\beta\gamma-1)}{(\beta +\lambda (x^*)^2)^2} = \omega^*.
\end{equation*}
We therefore obtain the following linear approximation (using Taylor's theorem) of $\phi_i$ around $x^*$. There exists $C$ such that for any sufficiently small $\tau>0$ the following is true. Suppose $R_i = x^*+y_i$ and $R = x^* + y$
where $|y_i|\leq\tau$ and $|y|\leq \tau$. Then
\begin{equation}\label{bou1}
| \phi_i(R) - (x^* + \omega^*(y_i + y) )| \leq C \tau^2
\end{equation}
and
\begin{equation}\label{bou2}
| \phi_i'(R) - \omega^* | \leq C |\tau|.
\end{equation}
We will choose $\tau>0$ small enough so that
\begin{equation}\label{tab}
\tau < \frac{|\omega^*|(1-|\omega^*|)}{100 C}.
\end{equation}
Recall that $|\omega^*|<1$ and hence~\eqref{bou2} implies that for every $\tau>0$ small enough
we have that for $R_i\in [x^*-\tau,x^*+\tau]$ we have that the map $\phi_i$ on the interval 
\begin{equation}\tag{\ref{idef2}}
I' = \Big[x^* - 2 \tau \frac{|\omega^*|}{1-|\omega^*|}, x^* + 2 \tau \frac{|\omega^*|}{1-|\omega^*|}\Big]
\end{equation} 
is uniformly contracting, using also Lemma~\ref{lem1}.  The bound on $\omega(R)$ follows from $\omega(x^*)<1$ (see Lemma~\ref{lem1}). This proves the first two inequalities in the lemma.

For the third inequality, consider $R_i\in [x^*-\tau,x^*+\tau]$ and let $R_i = x^* + y_i$, so that $|y_i|\leq\tau$.
Let also $R= x^* + y$; then, since $R\in I'$, we have $|y|\leq 2 \tau \frac{|\omega^*|}{1-|\omega^*|}$.

From~\eqref{bou1} we have
$$
|\phi_i(R) - x^*| \leq C \tau^2 + \left(\tau + 2 \tau \frac{|\omega^*|}{1-|\omega^*|}\right)|\omega^*|.
$$
To prove the lemma it is enough to show
\begin{equation}\label{mapr}
C \tau^2 + \left(\tau + 2 \tau \frac{|\omega^*|}{1-|\omega^*|}\right)|\omega^*| \leq 2 \tau \frac{|\omega^*|}{1-|\omega^*|}.
\end{equation}
Equation~\eqref{mapr} is equivalent to $C\tau^2 \leq 2\tau |\omega^*|$, which follows from our choice of $\tau$ in~\eqref{tab}.
\end{proof}

\begin{proof}[Proof of Lemma~\ref{lc}]
Let $x=(x_1+x_2)/(2\omega)$. Note that $x\in [x^*-\tau,x^*+\tau]$. Let $R_i$ be the effective field of the field gadget  guaranteed by Lemma~\ref{led},
that is, $|R_i - x|\leq \tau\delta$. We will show
\begin{equation}\label{eop}
\phi_i( x^* + |\omega|\tau/2) < x_1 \quad\mbox{and}\quad x_2 < \phi_i( x^* - |\omega|\tau/2).
\end{equation}
Using~\eqref{bou1} we have
\begin{align}
\phi_i (x^* + |\omega|\tau/2) - x_1 \leq x^* + \omega( (R_i-x^*) + |\omega|\tau/2 ) - x_1 + C|\tau|^2 \leq \\
x^* + |\omega| \tau\delta + \omega( (x-x^*) + |\omega|\tau/2 ) - x_1 + C|\tau|^2 \leq 2 |\omega| \tau\delta +
\omega |\omega|\tau/2 + C|\tau|^2 <0,
\end{align}
the last inequality is true for $\tau$ satisfying~\eqref{tab} and $\delta<|\omega|/100$  (the negative middle term ``overwhelms'' the remaining positive terms). The second inequality in~\eqref{eop} follows the same argument yielding
\[\phi_i (x^* - |\omega|\tau/2) - x_2 \geq - 2 |\omega| \tau\delta -
\omega |\omega|\tau/2 - C|\tau|^2 > 0.\qedhere\]
\end{proof}

\subsection{Proof of Lemma~\ref{case2}}\label{sec:case2}
In this section, we consider Case 2 and prove Lemma~\ref{case2}. We first argue that the maps ${\cal F}_{x,t}$ are getting ``flat'' as $t$ goes to infinity.

\begin{lemma}\label{shrink}
There exist constants $C<1,D<1$ and $C'>0,D'>0$ such that for any $x\in I$ any $t\geq 0$ and any $(f,g)\in{\cal F}_{x,t}$
we have
\begin{equation}\label{ee1}
\max_{R,R'\in I'} | f(R) - f(R') | \leq C' C^t.
\end{equation}
and
\begin{equation}\label{ee2}
\max_{R,R'\in I', M,M'\in J} | g(R,M) - g(R',M') |\leq D' D^t.
\end{equation}
\end{lemma}

\begin{proof}
We will prove the lemma by induction on $t$. Consider $(f,g)\in {\cal F}_{x,t}$. Let $i$ be such that
$x\in\phi_i^{-1}(I)$ and $(\hat{f},\hat{g})\in {\cal F}_{\phi_i^{-1}(x),t-1}$ be such that
$$f(R) = \phi_i(\hat{f}(R))\quad\mbox{and}\quad g(R,M)=\psi_i(\hat{f}(R),\hat{g}(R,M)).$$
Note that $\hat{f}(R)\in I'$ and $\hat{f}(R')$ (see Lemma~\ref{contr}).
Recall that~$\phi_i$ is contracting on $I'$ (see Lemma~\ref{contr}), hence
\begin{equation}\label{ttty}
| f(R) - f(R') | = |\phi_i(\hat{f}(R)) - \phi_i(\hat{f}(R'))| \leq C | \hat{f}(R) - \hat{f}(R') | \leq C' C^{t}.
\end{equation}
We have
$$
g(R,M) = 1-\omega(f(R))\Big( \hat{g}(f(R),M) + M_i - 1\Big)\quad\mbox{and}\quad
g(R',M') = 1-\omega(f(R'))\Big( \hat{g}(f(R'),M') + M_i - 1\Big).
$$
Hence
\begin{equation}\label{uuu7}
\begin{aligned}
g(R,M)- g(R',M') &= \big(\omega(f(R'))-\omega(f(R))\big)\big(M_i-1 + \hat{g}(f(R'),M')\big)\\
&\hskip 3.5cm  + \omega(f(R)) \Big(\hat{g}(f(R'),M')-\hat{g}(f(R),M)\Big).
\end{aligned}
\end{equation}
We will use~\eqref{ttty} to bound the first term in the final expression of~\eqref{uuu7}; the second
term will be bounded by induction. Note that the function $\omega(x)$ on $I'$ has a bounded derivative. Further note that
$M_i$ and $\hat{g}(f(R'),M')$ are in $J'$. Hence there exists a constant $K$ such that
\begin{equation}\label{ooooq}
\|(\omega(f(R'))-\omega(f(R)))(M_i-1 + \hat{g}(f(R'),M'))\| \leq K |f(R')-f(R) \leq K C' C^t.
\end{equation}
We can now take $D>\max\{C,C_{\mathrm{max}}\}$ and $D' = K C' D/(D-C_{\mathrm{max}})$. Then we have
\[| g(R,M)- g(R',M') | \leq K C' C^t + Q D' D^{t-1} \leq D' D^t.\qedhere\]
\end{proof}

We can now prove Lemma~\ref{yyyqw}.

\begin{proof}[Proof of Lemma~\ref{yyyqw}]
The lemma follows from~\eqref{ee1} and the fact that $x\in f(I)$. 
\end{proof}

\begin{proof}[Proof of Lemma~\ref{case2}]
Let $L_x$ be the intersection of $g(I'\times J)$ for all $(f,g)\in {\cal F}_{x,t}$ for all $t\geq 0$.
By Helly's theorem $L_x$ is non-empty. By Lemma~\ref{shrink} we have that $L$ consists of a single point
$\ell$. Let $F(x)=\ell$. Lemma~\ref{shrink} also implies~\eqref{close} for every $R\in I'$ and every $M\in J$
and $t$ large enough.  We have the following observation
about the function $F$. If $\phi_i(y)=x$ and $x,y\in I$ then
\begin{equation}\label{eq4}
F(x) = 1 - \omega(x) ( F(y) + M_i - 1).
\end{equation}
To prove~\eqref{eq4} note that if we take $(f,g)\in{\cal F}_{y,t}$ then we have
$(f_1,g_1)\in{\cal F}_{x,t+1}$ where
$$g_1(R,M) = \psi_i(f(R), g(R,M)) = 1 -  \omega(\phi_i(f(R)))(g(R,M) + M_i - 1).$$
Note that as $t$ goes to infinity we have that $f(R)$ goes to $y$, $\phi_i(f(R))$
goes to $x$, $g(R,M)$ goes to $F(y)$ and $g_1(R,M)$ goes to $F(x)$ (for any $R\in I'$
and $M\in J$) and hence~\eqref{eq4} holds.

Now we argue continuity of $F$. We will show that for $x_1,x_2\in I$ we have for some $c\in (0,1)$
\begin{equation}\label{equ2}
|F(x_1) - F(x_2)|\leq  C |x_1 - x_2|^{c}.
\end{equation}
We will prove~\eqref{equ2} by induction on $\lfloor\log_{C_{\mathrm{max}}} |x_1-x_2| \rfloor$,
where $C_{\mathrm{max}}$ is the constant from Lemma~\ref{contr}. The base case is when $|x_1-x_2|\geq |\omega|\tau\delta/2$. Since $F(x_1),F(x_2)\in J$
we have
$$|F(x_1)-F(x_2)|\leq 2 T \leq C |\omega|\tau\delta/2 \leq C (|\omega|\tau\delta/2)^c \leq C |x_1-x_2|^c,$$
where we choose $C\geq 4T/(|\omega|\tau\delta)$.

Now we proceed to the induction step. Let $x_1,x_2$ be such that $|x_1-x_2|\leq |\omega|\tau\delta/2$.
By Lemma~\ref{lc} there exists $\phi_i$ such that both $x_1,x_2$ are in $\phi_i(I)$. Let $y_1=\phi_i^{-1}(x_1)$
and $y_2=\phi_i^{-1}(x_2)$. By Lemma~\ref{contr} we have
$$(1/C_{\mathrm{max}})|x_1 - x_2| \leq |y_1-y_2|\leq (1/C_{\mathrm{min}})|x_1 - x_2|.$$
Note that
$$\lfloor\log_{C_{\mathrm{max}}} |y_1-y_2| \rfloor < \lfloor\log_{C_{\mathrm{max}}} |x_1-x_2| \rfloor.$$
From equation~\eqref{eq4} we have
\begin{equation}\label{eq6}
F(x_1) - F(x_2) = (F(y_1) + M_i - 1 ) (\omega(x_2)-\omega(x_1)) - \omega(x_2) (F(y_1) - F(y_2)).
\end{equation}
Note that the function $\omega(x)$ on $I'$ has a bounded derivative. Further note that
$M_i$ and $F(y_1)$ are in $J'$. Hence there exists a constant $K$ such that
\begin{equation}\label{eq7}
| (F(y_1) + M_i - 1 ) (\omega(x_2)-\omega(x_1)) | \leq K |x_2 - x_1|.
\end{equation}
Hence
\begin{equation}\label{ttyyy}
|F(x_1) - F(x_2)|\leq K |x_2 - x_1| + C_{\mathrm{max}} C |y_2 - y_1|^c \leq
\left( K  + C \frac{C_{\mathrm{max}}}{C_{\mathrm{min}}^c}\right)  |x_2 - x_1|^c.
\end{equation}
We will choose $c\in (0,1)$ such that $C_{\mathrm{min}}^c > C_{\mathrm{max}}$ (such $c\in (0,1)$
exists since $C_{\mathrm{min}}, C_{\mathrm{max}} \in (0,1)$. We then choose $C$
such that $C (\frac{C_{\mathrm{max}}}{C_{\mathrm{min}}^c}-1) + K \leq 0$. For
our choice of $c\in (0,1)$ and $C$ we can bound~\eqref{ttyyy} by $C |x_2 - x_1|^c$
completing the induction step.
\end{proof}

\subsection{Obtaining an infinite sequence of pairs of field gadgets}\label{sec:vanilla}
In this section, we prove Lemma~\ref{lem:vanilla}.
\begin{proof}[Proof of Lemma~\ref{lem:vanilla}]
For $j=0$, we set $\Tc_{1,0}=\Tc_1,\Tc_{2,0}=\Tc_2$. For $j\geq 1$, we proceed inductively. For $i\in \{1,2\}$, we construct $\Tc_{i,j+1}$ by just combining $\Tc_{i,j}$ with the degenerate tree of  Example~\ref{eee} according to Lemma~\ref{lem1}; we then have that
\[R_{i,{j+1}}=\frac{1+\gamma\lambda R_{i,j}}{\beta+\lambda R_{i,j}}, \qquad M_{i,j+1}  = 1 - \omega(R_{i,j}) M_{i,j}.\]
Since $R_{1,j}=R_{2,j}$ and $M_{1,j}\neq M_{2,j}$, we therefore have that $R_{1,{j+1}}=R_{2,{j+1}}$ and $M_{1,j+1}\neq M_{2,j+1}$ as well. The function $f(x)=\frac{1+\gamma\lambda x}{\beta+\lambda x}$ has a unique fixpoint $x^*$ in the interval $(\beta,1/\gamma)$ and  for any starting point $x_0\in (\beta,1/\gamma)$ with $x_0\neq x^*$, repeated application of $f$ produces $x_i$'s that are getting strictly closer to $x^*$. It follows that the sequence $\{R_{1,j}\}$ constructed above, and hence $\{R_{2,j}\}$ as well, has pairwise distinct elements except for the case that $R_{1}=R_{2}=x^*$.  In this case, we let $\Tc_{i,0}$ to be the field gadget obtained  by combining two copies of  $\Tc_i$ according to Lemma~\ref{lem1}  and proceed as before. Note that whenever $x^*\neq 1$, we also have that $R_{i,0}\neq x^*$ as needed since \[R_{i,0} =\frac{1+\gamma\lambda (x^*)^2}{\beta+\lambda (x^*)^2 }\neq f(x^*)=x^*,\]
where the disequality holds because $x^*\neq 1$. We also have that $R_{1,0}=R_{2,0}$ and $M_{1,0}\neq M_{2,0}$ since, for $i\in \{1,2\}$, $M_{i,0}=1-\omega(R_{i,0})(2M_{i}-1)$ and $M_1\neq M_2$.   Finally, in the case that $x^*=1$, which corresponds to $\lambda=\frac{1-\beta}{1-\gamma}$ for $\beta,\gamma\in(0,1)$ and $\beta\neq \gamma$, we let $\Tc_{i,0}$ to be the field gadget obtained  by combining,  according to Lemma~\ref{lem1}, $\Tc_i$ and the field gadget $\Tc$ of Lemma~\ref{lem:triangle}. Then, we have that  $R_{i,0}\neq x^*$ as needed since \[R_{i,0} =\frac{1+\gamma\lambda R}{\beta+\lambda R }\neq 1,\]
since the effective field $R$ of $\Tc$ satisfies $R\neq 1$ by Lemma~\ref{lem:triangle}.  We also have that $R_{1,0}=R_{2,0}$ and $M_{1,0}\neq M_{2,0}$ since, for $i\in \{1,2\}$, $M_{i,0}=1-\omega(R_{i,0})(M_{i}+M-1)$ and $M_1\neq M_2$, where $M$ is the magnetization gap of $\Tc$. 

This finishes the proof of Lemma~\ref{lem:vanilla}.
\end{proof} 

\section{Proofs of the remaining inapproximability results}
With all the field-gadget constructions in place, we now give/finish the proofs of the remaining inapproximability results.

\subsection{The bipartite ``phase-gadget'' of Sly and Sun}\label{sec:SlySun}
Here, we state more precisely the properties of the bipartite phase-gadgets that we described in Section~\ref{sec:phase-gadget}. 
Recall that a graph $G\in \Gc^{\ell}_{n}$ is an almost-$\Delta$-regular bipartite graph with $n$ vertices on each side and $\ell$ ``ports'' of degree $\Delta-1$, 
which will be used to connect distinct copies of gadgets. Recall also that for $\sigma: U^\pl\cup U^\mi\rightarrow \{0,1\}$, we define the \emph{phase} $\Yc(\sigma)$ of the configuration $\sigma$ as the side which has the most occupied vertices under $\sigma$, i.e., 
\[\Yc(\sigma)=\begin{cases} \pl & \mbox{ if } |\sigma_{U^\pl}|\geq |\sigma_{U^\mi}|,\\ \mi& \mbox{otherwise}. \end{cases}\]

\begin{lemma}[{\cite[Section 4.1]{SlySun}}]\label{lem:SlySun}
Let $\Delta\geq 3$ and $(\beta,\gamma,\lambda)\in \Uc_{\Delta}$. There exist real numbers $q^\pl, q^\mi\in (0,1)$ with $q^\mi<q^\pl$ such that the following hold for every integer $\ell\geq 0$, real $\eps>0$, and all sufficiently large integers $n$. There exists a graph $G\in \Gc^{\ell}_{n}$ such that the Gibbs distribution $\mu=\mu_{G;\beta,\gamma, \lambda}$ satisfies the following:
\begin{enumerate}
\item \label{it:balancedphases} $\tfrac{1-\eps}{2}\leq\mu(\Yc(\sigma)=\plm)\leq \tfrac{1+\eps}{2}$.
\item \label{it:approxindep}for every $\tau: W\rightarrow \{0,1\}$, we have that $(1-\eps)Q_W^{\plm}(\tau)\leq \mu(\sigma_W=\tau\mid \Yc(\sigma)=\plm)\leq (1+ \eps)Q_W^{\plm}(\tau)$, where $Q^\pl_W(\cdot), Q^\mi_W(\cdot)$ are product distributiond defined by 
\begin{equation}\label{eq:product}
Q^\plm_W(\tau)= (q^\plm)^{|\tau^{-1}(1) \cap W^\pl|}(1-q^\plm)^{|\tau^{-1}(0) \cap W^\pl|}(q^\mip)^{|\tau^{-1}(1) \cap W^\mi|}(1-q^\mip)^{|\tau^{-1}(0) \cap W^\mi|}.
\end{equation}
\end{enumerate}
\end{lemma}

\begin{remark}\label{rem:gadgetcompute}
We remark that for any fixed $\eps>0$ and $\ell\geq 1$, we can find $n$ and $G\in \Gc^{\ell}_{n}$ satisfying Items~\ref{it:balancedphases} and~\ref{it:approxindep} by brute force in constant time. Namely, the  values of $q^\pl,q^\mi$ can be obtained by solving the system of equations $x=\frac{1}{\lambda}\left(\frac{ \beta y+1}{y+\gamma}\right)^{\Delta-1},y=\frac{1}{\lambda}\left(\frac{ \beta x+1}{x+\gamma}\right)^{\Delta-1}$ for $x,y>0$; it can be shown that for $(\beta,\gamma,\lambda)\in \Uc_{\Delta}$ this system has a unique solution with $x<y$ (see, e.g., \cite[Lemma 7]{GSVIsing}, or \cite[Section 6.2]{MSW07}). Using these values, for any constants $n,\ell$ and any graph $G\in \Gc^{\ell}_{n}$ we can check whether Items~\ref{it:balancedphases} and~\ref{it:approxindep} hold in constant time, and therefore find $G$ satisfying Lemma~\ref{lem:SlySun}.
\end{remark}

\subsection{Preliminaries}
Let $G=(V,E)$ be a graph and $(\beta,\gamma)$ be antiferromagnetic. It will sometimes be convenient to consider models where instead of a uniform field $\lambda$ over vertices of $G$, we allow vertices to have their ``own'' field. More precisely, for a field vector $\lambdab=\{\lambda_v\}_{v\in V}$, we define 
\begin{equation}\label{eq:multi}
\mu_{G;\beta,\gamma,\lambdab}(\sigma)=\frac{w_{G;\beta,\gamma,\lambdab}(\sigma)}{Z_{G;\beta,\gamma,\lambdab}}, \mbox{ where \ \  }w_{G;\beta,\gamma,\lambdab}(\sigma):=\beta^{m_{0}(\sigma)}\gamma^{m_{1}(\sigma)}\prod_{v\in V;\, \sigma(v)=1} \lambda_v.
\end{equation}

\begin{lemma}\label{lem:perturb}
Let $(\beta,\gamma)$ be antiferromagnetic and $\lambda,\lambda_1,\lambda_2>0$. Let $G=(V,E)$ be a graph and $S\subseteq V$. For $i\in \{1,2\}$, let $\lambdab_i$ be the field vector on $V$, where every $v\in S$ has field $\lambda_i$, whereas every $v\in V\backslash S$ has field $\lambda$. Let $\mu_i$ be the Gibbs distributions on $G$, respectively,  with parameters $\beta,\gamma, \lambdab_i$. Then, for every $v\in V$,  it holds that 
\[\big|\Eb_{\sigma\sim \mu_2}[\sigma(v)]-\Eb_{\sigma\sim \mu_1}[\sigma(v)]\big|\leq 2|S|\, \Big|\frac{\lambda_2}{\lambda_1}-1\Big|.\]
\end{lemma}
\begin{proof}
Assume first that $\lambda_1\geq \lambda_2$. For $v\in V$, let 
\[\Zin_{G,v}(\lambdab_i):= \sum_{\sigma:V\rightarrow \{0,1\}; \, \sigma(v)=1} w_{G;\beta,\gamma,\lambdab_i}(\sigma),\quad Z_{G}(\lambdab_i):= \Zin_{G;\beta,\gamma,\lambdab_i},\]
so that $\Eb_{\sigma\sim \mu_i}[\sigma(v)]=\frac{\Zin_{G,v}(\lambdab_i)}{Z_{G}(\lambdab_i)}$. Therefore,
\begin{align}\label{eq:4tg4ggteewe}
\Eb_{\sigma\sim \mu_1}[\sigma(v)]-\Eb_{\sigma\sim \mu_2}[\sigma(v)]=\frac{\Zin_{G,v}(\lambdab_1)-\Zin_{G,v}(\lambdab_2)}{Z_{G}(\lambdab_1)}+\frac{\Zin_{G,v}(\lambdab_2)}{Z_{G}(\lambdab_2)}\frac{Z_G(\lambdab_2)-Z_{G}(\lambdab_1)}{Z_{G}(\lambdab_1)}
\end{align}
For arbitrary $\sigma: V\rightarrow \{0,1\}$, we have that
\[\frac{w_{G;\beta,\gamma,\lambdab_1}(\sigma)-w_{G;\beta,\gamma,\lambdab_2}(\sigma)}{w_{G;\beta,\gamma,\lambdab_1}(\sigma)}=\frac{\lambda_1^{|\sigma_S|}- \lambda_2^{|\sigma_S|}}{\lambda_1^{|\sigma_S|}}\leq |S|\Big(1-\frac{\lambda_2}{\lambda_1}\Big).\]
From this, it follows that  
\[\Big|\frac{\Zin_{G,v}(\lambdab_1)-\Zin_{G,v}(\lambdab_2)}{Z_{G}(\lambdab_1)}\Big|\leq |S|\Big(1-\frac{\lambda_2}{\lambda_1}\Big), \quad \Big|\frac{Z_G(\lambdab_2)-Z_{G}(\lambdab_1)}{Z_{G}(\lambdab_1)}\Big|\leq |S|\Big(1-\frac{\lambda_2}{\lambda_1}\Big).\]
Hence, by taking absolute values in \eqref{eq:4tg4ggteewe} and using $\frac{\Zin_{G,v}(\lambdab_2)}{Z_{G}(\lambdab_2)}\leq 1$ and the triangle inequality, we obtain
\[\big|\Eb_{\sigma\sim \mu_1}[\sigma(v)]-\Eb_{\sigma\sim \mu_2}[\sigma(v)]\big|\leq 2|S| \Big(1-\frac{\lambda_2}{\lambda_1}\Big),\]
finishing the proof for $\lambda_1\geq \lambda_2$. The case $\lambda_1<\lambda_2$ follows from the previous case since
\[\big|\Eb_{\sigma\sim \mu_1}[\sigma(v)]-\Eb_{\sigma\sim \mu_2}[\sigma(v)]\big|\leq 2|S| \Big(1-\frac{\lambda_1}{\lambda_2}\Big)\leq 2|S| \Big(\frac{\lambda_2}{\lambda_1}-1\Big).\qedhere\]
\end{proof}

\subsection{Proof of Lemma~\ref{lem:maingadget}}\label{sec:maingadget}
We start with the proof of Lemma~\ref{lem:maingadget}.
\begin{proof}[Proof of Lemma~\ref{lem:maingadget}]
We first prove the bounds on $\Mc_{\beta,\gamma,\lambda}(H^{k}_{G,\Tc})$. 

For an edge $e=(u,v)\in E_H$ and $i\in [k]$, let $P^{i,\pl}_{e},P^{i,\mi}_{e}$ be the $i$-th paths connecting the $\pl/\pl$ and $\mi/\mi$ sides of the gadgets $G_u,G_v$, respectively. Let the (two) internal vertices of the path $P^{i,\plm}_{e}$ be $t^{i,\plm}_{e,1}$ and $t^{i,\plm}_{e,2}$. For $j=1,2$ denote by $\Tc^{i,\plm}_{e,j}$ the copy of $\Tc$ that was appended  to $t^{i,\plm}_{e,j}$. 

Note that the vertex set of $H^{k}_{G,\Tc}$ is the disjoint union of $V(\widehat{H}^{k}_{G})$  and $V(\Tc^{i,s}_{e,j})$ for $e\in E(H), s\in\{\pl,\mi\}, i\in [k], j\in \{1,2\}$, so we have 
\begin{equation}\label{eq:MHkGT}
\Mc_{\beta,\gamma,\lambda}(H^{k}_{G,\Tc})=\Eb_{\sigma\sim \mu}\big[\big|\sigma_{V(\widehat{H}^{k}_{G})}\big|\big]+\sum_{e\in E(H)}\sum_{s\in\{\pl,\mi\}}\sum_{i\in [k]}\sum_{j\in \{1,2\}}\Eb_{\sigma\sim \mu}\Big[\big|\sigma_{V(\Tc^{i,s}_{e,j})}\big|\Big].
\end{equation}
Fix arbitrary $e\in E(H), s\in\{\pl,\mi\}, i\in [k], j\in \{1,2\}$, and for convenience let $\mu^{i,s}_{e,j}$ denote the Gibbs distribution on $\Tc^{i,s}_{e,j}$ with parameters $\beta,\gamma,\lambda$. Note that,  conditioned on the spin of $t^{i,s}_{e,j}$, the distribution $\mu$ factorizes; more precisely, for $\tau: V(\Tc^{i,s}_{e,j})\rightarrow \{0,1\}$ and $w\in \{0,1\}$, we have
\[\mu\Big(\sigma_{V(\Tc^{i,s}_{e,j})}=\tau\mid \sigma(t^{i,s}_{e,j})=w\Big)=\mu^{i,s}_{e,j}(\tau\mid \tau(t^{i,s}_{e,j})=w).\] 
Using that the magnetization gap of $\Tc^{i,s}_{e,j}$ is $M$, we therefore obtain that 
\begin{align*}
\Eb_{\sigma\sim \mu}\Big[\big|\sigma_{V(\Tc^{i,s}_{e,j})}\big|\Big]&=\sum_{w\in\{0,1\}}\mu\big(\sigma(t^{i,s}_{e,j}\big)=w)\, \Eb_{\tau\sim \mu^{i,s}_{e,j}}\big[\,|\tau|\mid \tau\big(t^{i,s}_{e,j}\big)=w\big]\\
&=M\,\Eb_{\sigma\sim \mu}[\sigma(t^{i,s}_{e,j})]+\Ac',
\end{align*}
where the last equality follows by writing $\mu\big(\sigma(t^{i,s}_{e,j})=0\big)=1-\mu\big(\sigma(t^{i,s}_{e,j})=1\big)=1-\Eb_{\sigma\sim \mu}[\sigma(t^{i,s}_{e,j})]$ and noting that $\Ac'=\Eb_{\tau\sim \mu^{i,s}_{e,j}}[\,|\tau|\mid \tau\big(t^{i,s}_{e,j}\big)=0]$.
Plugging into \eqref{eq:MHkGT}, we obtain
\begin{equation}\label{eq:4g4t4g56yttere5}
\Mc_{\beta,\gamma,\lambda}(H^{k}_{G,\Tc})=4k\Ac'|E(H)|+\Eb_{\sigma\sim \mu}\big[\big|\sigma_{V(\widehat{H}^{k}_{G})}\big|\big]+M \sum_{e\in E(H)}\sum_{s\in\{\pl,\mi\}}\sum_{i\in [k]}M^{i,s}_{e},
\end{equation}
where, for $e\in E(H)$, $i\in [k]$ and $s\in \{\pl,\mi\}$,
\[M^{i,s}_{e}:=\Eb_{\sigma\sim \mu}[\sigma(t^{i,s}_{e,1})+\sigma(t^{i,s}_{e,2})].\]
We can expand $M^{i,s}_{e}$ according to the configuration  on the ports $W^{i,s}_e:= \{w^{i,s}_{u},w^{i,s}_{v}\}$,   i.e, 
\begin{equation}\label{eq:4tg5tb5yhyhnh}
M^{i,s}_{e}=\sum_{\tau: W^{i,s}_e\rightarrow \{0,1\}} \mu\big(\sigma_{W^{i,s}_e}=\tau\big) \Eb_{\sigma\sim \mu}[\sigma(t^{i,s}_{e,1})+\sigma(t^{i,s}_{e,2})\mid \sigma_{W^{i,s}_e}=\tau].
\end{equation}
For convenience, let $\overline{W}^{i,s}_e$ be the set $(W_u\cup W_v)\backslash W^{i,s}_e$, i.e., the set of all ports in the gadgets $G_u,G_v$ other than those in $W^{i,s}_e$. We have 
\begin{equation}\label{eq:rtb5b5hh5}
\mu\big(\sigma_{W^{i,s}_e}=\tau\big)=\sum_{\substack{Y: V(H)\rightarrow \{\pl,\mi\};\\\eta:\overline{W}^{i,s}_e\rightarrow \{0,1\}}}\mu\big(\sigma_{W^{i,s}_e}=\tau\mid \widehat{\Yc}(\sigma)=Y,\sigma_{\overline{W}^{i,s}_e}=\eta\big)\mu\big(\widehat{\Yc}(\sigma)=Y,\sigma_{\overline{W}^{i,s}_e}=\eta\big).
\end{equation}
Now, observe that $\overline{W}^{i,s}_e$ disconnects the vertices in $U_u\cup U_v$ from the rest of the graph $H^{k}_{G,\Tc}$, so, conditioned on $\eta:\overline{W}^{i,s}_e\rightarrow \{0,1\}$, the configuration on $W^{i,s}_e$ is distributed according to $\mu_e$, where $\mu_e$ is the Gibbs distribution with parameters $(\beta,\gamma,\lambda)$ on the subgraph of $H^{k}_{G,\Tc}$ induced by $U_u\cup U_v$. Therefore, for any $Y: V(H)\rightarrow \{\pl,\mi\}$ and $\eta:\overline{W}^{i,s}_e\rightarrow \{0,1\}$, we have that
\begin{equation}\label{eq:6N6J8R}
\mu\big(\sigma_{W^{i,s}_e}=\tau\mid \widehat{\Yc}(\sigma)=Y,\sigma_{\overline{W}^{i,s}_e}=\eta\big)=\mu_e\big(\sigma_{W^{i,s}_e}=\tau\mid \Yc(\sigma_{U_u})=Y_u, \Yc(\sigma_{U_v})=Y_v,\sigma_{\overline{W}^{i,s}_e}=\eta).
\end{equation}
Since $G_u,G_v$ are copies of $G$, and $G$ satisfies Lemma~\ref{lem:SlySun}, by Item~\ref{it:approxindep} we obtain that $\mu_e\big(\sigma_{W^{i,\plm}_e}=\tau\mid \Yc(\sigma_{U_u})=Y_u, \Yc(\sigma_{U_v})=Y_v,\sigma_{\overline{W}^{i,s}_e}=\eta)$ is approximately independent of $\eta$ and  within a factor of $(1\pm 8\epsilon)$ from the Gibbs distribution on $P^{i,\plm}_e$ with field vector $\lambdab_{\plm}(Y_u,Y_v)$ given by 
\[\lambda_{w^{i,\plm}_u}=\tfrac{q^{\plm Y_u}}{1-q^{\plm Y_u}},\quad \lambda_{t^{i,\plm}_{e,1}}=\lambda_{t^{i,\plm}_{e,2}}=\lambda R,\quad \lambda_{w^{i,\plm}_u}=\tfrac{q^{\plm Y_v}}{1-q^{\plm Y_v}}, \mbox{ with the notation  $\plm \plm\equiv\plm$ and $\plm \mip\equiv\mip$}.\] 
More precisely, we have 
\begin{equation}\label{eq:tb5b5h5646}
\mu_e\big(\sigma_{W^{i,s}_e}=\tau\mid \Yc(\sigma_{U_u})=Y_u, \Yc(\sigma_{U_v})=Y_v,\sigma_{\overline{W}^{i,s}_e}=\eta)=(1\pm 8\epsilon)\mu_{P^{i,s}_e; \beta,\gamma,\lambdab_{s}(Y_u,Y_v)}(\sigma_{W^{i,\plm}_e}=\tau).
\end{equation}
Combining \eqref{eq:6N6J8R} and \eqref{eq:tb5b5h5646}, and plugging back into \eqref{eq:rtb5b5hh5}, we obtain
\[\mu\big(\sigma_{W^{i,s}_e}=\tau\big)=(1\pm 8\epsilon)\sum_{Y: V(H)\rightarrow \{\pl,\mi\}}\mu\big(\widehat{\Yc}(\sigma)=Y\big)\, \mu_{P^{i,s}_e; \beta,\gamma,\lambdab_{s}(Y_u,Y_v)}(\sigma_{W^{i,s}_e}=\tau)\]
In turn, plugging this back into \eqref{eq:4tg5tb5yhyhnh} yields that
\begin{equation}\label{eq:45g54trexxa}
M^{i,s}_{e}=(1\pm 8\epsilon)\sum_{Y: V(H)\rightarrow \{\pl,\mi\}}\mu\big(\widehat{\Yc}(\sigma)=Y\big)\Eb_{\sigma\sim \mu_{P^{i,s}_e; \beta,\gamma,\lambdab_{s}(Y_u,Y_v)}}[\sigma(t^{i,s}_{e,1})+\sigma(t^{i,s}_{e,2})].
\end{equation}

For $s_1,s_2,s_3\in\{\pl,\mi\}$, let
\[A^{s_1}_{s_2s_3}:=\Eb_{\sigma\sim \mu_{P^{i,s_1}_e; \beta,\gamma,\lambdab_{s}(s_2,s_3)}}[\sigma(t^{i,s_1}_{e,1})+\sigma(t^{i,s_1}_{e,2})].\]
To give expressions for $A^{s_1}_{s_2s_3}$, it will be convenient to consider the following vectors/matrices:
\begin{equation}\label{eq:4t4vt4ggyh}
\qb^{\plm}:=\Big[\begin{smallmatrix} 1-q^\plm\\ q^\plm  \end{smallmatrix}\Big], \quad \Lb(x):=\Big[\begin{smallmatrix} \beta^3+2\beta x+\gamma x^2 & \beta^2 + x(1 + \beta \gamma) + \gamma^2 x^2 \\
\beta^2 + x(1 + \beta \gamma) + \gamma^2 x^2  & \beta + 2\gamma x+ \gamma^3 x^2 \end{smallmatrix}\Big], \quad \begin{array}{c} f_{\plm \plm}(x):=(\qb^{\plm})^{\T}\, \Lb(x) \qb^{\plm}\\ f_{\plm \mip}(x):=(\qb^{\plm})^{\T}\, \Lb(x) \qb^{\mip}\end{array}.
\end{equation}
where $q^{\pl},q^{\mi}$ are the constants in Lemma~\ref{lem:SlySun}, and recall that $q^{\pl}\neq q^{\mi}$. Roughly, the matrix $\Lb(x)$ corresponds to the interaction matrix between the endpoints of a path of length 4, whose vertices have vertex fields $1,x,x,1$ respectively. Then, the functions $f_{\plm \plm}(x)$ capture the total weight of configurations on the path, when the endpoints are occupied with probabilities $q^{\plm},q^{\plm}$, respectively; similarly for the functions  $f_{\plm \mip}(x)$. The magnetization of the two middle vertices is given by the log-derivative with respect to $x$ and multiplying by $x$ (this formula can be verified using the definition of magnetization), yielding the following expressions for the $A$'s: 
\begin{gather*}
A^{\pl}_{\pl\pl}=A^{\mi}_{\mi\mi}=\lambda R\cdot \tfrac{f_{\pl \pl}'(\lambda R)}{f_{\pl \pl}(\lambda R)},\quad  A^{\mi}_{\pl\pl}=A^{\pl}_{\mi\mi}=\lambda R\cdot \tfrac{f_{\mi \mi}'(\lambda R)}{f_{\mi \mi}(\lambda R)}, \\ 
A^{\pl}_{\pl\mi}=A^{\pl}_{\mi\pl}=A^{\mi}_{\mi\pl}= A^{\mi}_{\pl\mi}=\lambda R\cdot\tfrac{f_{\pl \mi}'(\lambda R)}{f_{\pl \mi}(\lambda R)}=\lambda R\cdot \tfrac{f_{\mi \pl}'(\lambda R)}{f_{\mi\pl}(\lambda R)}.
\end{gather*}
Now, let's fix a phase vector $Y:V(H)\rightarrow \{\pl,\mi\}$, and compute 
\begin{equation}\label{eq:34f5646y}
\sum_{e=(u,v)\in E(H)}\sum_{s\in\{\pl,\mi\}}\sum_{i\in [k]}\Eb_{\sigma\sim \mu_{P^{i,s}_e; \beta,\gamma,\lambdab_{s}(Y_u,Y_v)}}[\sigma(t^{i,s}_{e,1})+\sigma(t^{i,s}_{e,2})].
\end{equation}
where the functions $f_{\plm\plm},f_{\plm \mip}$ are as in \eqref{eq:4t4vt4ggyh}.
We have that the contribution from an edge $e=(u,v)$ with $Y_u\neq Y_v$ is $kB$, while the contribution from an edge $e=(u,v)$ with $Y_u= Y_v$ is $kC$, where 
\begin{align*}B&:=A^{\pl}_{\pl\mi}+A^{\mi}_{\pl\mi}=A^{\pl}_{\mi\pl}+A^{\mi}_{\mi\pl}=\lambda R\Big(\tfrac{f_{\pl \mi}'(\lambda R)}{f_{\pl \mi}(\lambda R)}+\tfrac{f_{\mi \pl}'(\lambda R)}{f_{\mi\pl}(\lambda R)}\Big),\\
 C&:=A^{\pl}_{\pl\pl}+A^{\mi}_{\pl\pl}=A^{\mi}_{\mi\mi}+A^{\pl}_{\mi\mi}=\lambda R\Big(\tfrac{f_{\pl \pl}'(\lambda R)}{f_{\pl \pl}(\lambda R)}+\tfrac{f_{\mi \mi}'(\lambda R)}{f_{\mi \mi}(\lambda R)}\Big).
\end{align*}
It follows that the value of the sum in \eqref{eq:34f5646y} is
\begin{equation}\label{eq:4tt6gtgrrerrv}
k B\, \Cut_H(Y)+k C\big(|E(H)|- \Cut_H(Y)\big)=k (B-C)\Cut_H(Y)+kC |E(H)|.
\end{equation}
Combining \eqref{eq:45g54trexxa} with \eqref{eq:4g4t4g56yttere5} and using the value of the sum in \eqref{eq:34f5646y} as obtained in \eqref{eq:4tt6gtgrrerrv},  finishes the proof of the bounds on $\Mc_{\beta,\gamma,\lambda}(H^{k}_{G,\Tc})$.

We next prove the bounds on $\AvgCut_{\mu}(H)$. The upper bound is trivial, so we focus on the lower bound, which is along the same lines as in \cite[Lemma 4.2]{SlySun}. Let $\hat{\mu}$ be the Gibbs distibution on $\widehat{H}^k_G$ and note that this is a product distribution over the gadgets. For convenience, let $\Lb=\{L_{ij}\}_{i,j\in\{0,1\}}$ denote the (entries of the) matrix $\Lb(\lambda R)$. For phase configurations $Y:V(H)\rightarrow \{\pl,\mi\}$, we have that 
\[\mu(\hat{\Yc}(\sigma)=Y)\propto \hat{\mu}(\hat{\Yc}(\sigma)=Y)\sum_{\tau:W\rightarrow \{0,1\}}\hat{\mu}(\sigma_W=\tau\mid \hat{\Yc}(\sigma)=Y)\prod_{e=(u,v)\in E(H)}\prod_{s\in \{\pl, \mi\}}\prod_{i\in [k]}L_{\tau(w^{i,s}_u),\tau(w^{i,s}_v)}\]
By Items~\ref{it:balancedphases} and~\ref{it:approxindep} of Lemma~\ref{lem:SlySun}, for all $Y:V(H)\rightarrow \{\pl,\mi\}$ and $\tau:W\rightarrow \{0,1\}$, it holds that
\[\hat{\mu}\big(\hat{\Yc}(\sigma)=Y\big)=\Big(\frac{1\pm \epsilon}{2}\Big)^{|V(H)|}, \quad \hat{\mu}(\sigma_W=\tau\mid \hat{\Yc}(\sigma)=Y)=(1\pm \epsilon)^{|V(H)|} \prod_{v\in V}Q^{Y_v}_{W_v}(\tau_{W_v}),\]
where, for $v\in V(H)$, $Q^{\plm}_{W_v}$ is the product distribution defined in \eqref{eq:product} for the gadget $G_v$. Using that 
\[\mu(\hat{\Yc}(\sigma)=Y)\propto \big(\tfrac{1}{2}\pm\epsilon\big)^{|V(H)|}\prod_{e=(u,v)\in E(H)}\prod_{s\in \{\pl,\mi\}}\prod_{i\in [k]}\sum_{\tau:W^{i,s}_{e}\rightarrow \{0,1\}}Q^{Y_u}(\tau(w^{i,s}_u))L_{\tau(w^{i,s}_u),\tau(w^{i,s}_v)}Q^{Y_v}_{W_v}(\tau(w^{i,s}_v)).\]
Observe now that $\sum_{\tau:W^{i,s}_e\rightarrow \{0,1\}}Q^{Y_u}_{W_v}(\tau(w^{i,s}_u))L_{\tau(w^{i,s}_u),\tau(w^{i,s}_v)}Q^{Y_v}(\tau(w^{i,s}_v))=(\qb^{Y_u})^{\T} \Lb\, \qb^{Y_v}$, so
\[\mu(\hat{\Yc}(\sigma)=Y)\propto \big(\tfrac{1}{2}\pm\epsilon\big)^{|V(H)|}\prod_{(u,v)\in E(H); Y_u\neq Y_v}\big(f_{\pl \mi}(\lambda R)f_{\mi \pl}(\lambda R)\big)^{k}\prod_{(u,v)\in E(H); Y_u=Y_v}\big(f_{\pl \pl}(\lambda R)f_{\mi \mi}(\lambda R)\big)^{k}.\]
It follows that for 
\[A:=A(R)=\frac{f_{\pl \mi}(\lambda R)\cdot f_{\mi \pl}(\lambda R)}{f_{\mi \mi}(\lambda R)\cdot f_{\pl \pl}(\lambda R)}, \]
we have that  
\[\mu(\hat{\Yc}(\sigma)=Y)\propto \big(\tfrac{1}{2}\pm\epsilon\big)^{|V(H)|}A^{k\Cut_H(Y)}.\]
Now note that for $R>0$, we have $A>1$ since $q^{\pl}\neq q^{\mi}$ and for all $x$ we have that
\[f_{\pl \mi}(x)\cdot f_{\mi \pl}(x)-f_{\pl \pl}(x)\cdot f_{\mi \mi}(x)=(1-\beta \gamma)^3x^2 (q^{\pl}-q^{\mi})^2.\]
Note also that $A(0)=1$, while $B(R)-C(R)=(\log A(R))'$ thus proving that $A(R),B(R),C(R)$ are rational functions of $R$ satisfying \eqref{eq:vv6gvgg}.

Let $\hat{K}>1$ be such that $1-\frac{1}{\hat{K}}\leq \frac{4}{k\log A}$. From the estimates above, we have
\begin{equation}\label{eq:45gg4fcfr}
\mu\Big(\frac{\Cut_H\big(\hat{\Yc}(\sigma)\big)}{\MaxCut(H)}\leq \frac{1}{\hat{K}}\Big)\leq \frac{\big(\tfrac{1}{2}+\epsilon\big)^{|V(H)|} 2^{|V(H)|} A^{ \frac{k}{\hat{K}}\MaxCut(H)}}{\big(\tfrac{1}{2}-\epsilon\big)^{|V(H)|}A^{k \MaxCut(H)}}\leq \Big(\frac{1}{2}\Big)^{|V(H)|},
\end{equation}
where in the last inequality we used that 
\[A,\hat{K}>1,\quad \MaxCut(H)\geq \tfrac{1}{2}|E(H)|=\tfrac{3}{4}|V(H)|,\mbox{  and  } 2\big(\tfrac{1}{2}+\epsilon\big)/\big(\tfrac{1}{2}-\epsilon\big)\leq 4.\]
Since $\AvgCut_{\mu}(H):=\sum_{Y:V(H)\rightarrow \{\pl,\mi\}}\mu\big(\widehat{\Yc}(\sigma)=Y\big)|\Cut_H(Y)|$, from \eqref{eq:45gg4fcfr}, and using that $k\geq 10/\log A$, for  $|V(H)|\geq 2kA$, we obtain
\[\AvgCut_\mu(H)\geq \big(1-(1/2)^{|V(H)|}\big)\frac{\MaxCut(H)}{\hat{K}}\geq \frac{\MaxCut(H)}{K},\]
where $K=1+\frac{6}{k\log A}$ is as in the statement of the lemma, as wanted.

This finishes the proof of Lemma~\ref{lem:maingadget}.
\end{proof}

\subsection{Constant-factor inapproximability --- Proof of Theorems~\ref{thm:mainind3},~\ref{thm:mainIs3},~\ref{thm:maingen3}}\label{sec:f4343c}
We now give the relatively straightforward proofs of  Theorems~\ref{thm:mainind3},~\ref{thm:mainIs3},~\ref{thm:maingen3}.

\begin{proof}[Proof of Theorems~\ref{thm:mainind3},~\ref{thm:mainIs3},~\ref{thm:maingen3}]
Analogously to the proof of Theorem~\ref{thm:mainind1}, these theorems can  be obtained as immediate corollaries of Theorem~\ref{thm:maingen3} using the field-gadget constructions of Theorem~\ref{thm:const2}.  The only extra fact that we need to note is that in the case of the hard-core model ($\beta=1,\gamma=0$), we have that $(\beta,\gamma,\lambda)\in \Uc^*_\Delta$ iff $\lambda>\lambda_c(\Delta)$, for all $\Delta\geq 3$. Similarly in the case of the antiferromagnetic Ising model ($0<\beta=\gamma<1$), we have that $(\beta,\gamma,\lambda)\in \Uc^*_\Delta$ iff $\lambda\in  (\tfrac{1}{\lambda_c},\lambda_c)$, where $\lambda_c=\lambda^{\Ising}_c(\Delta,\beta)$, for all $\Delta\geq 3$. 
\end{proof}

\subsection{Proof of Theorems~\ref{thm:mainind2},~\ref{thm:mainIs2}, and~\ref{thm:maingen2}}\label{sec:maingen2}
We next give the proof of Theorem~\ref{thm:maingen2}, using Theorem~\ref{thm:occgadget}.
\begin{proof}[Proof of Theorem~\ref{thm:maingen2}]
Suppose for the sake of contradiction that, for arbitrarily small $\kappa>0$, there is a polynomial-time algorithm that, on input a graph $G$ of maximum degree $\Delta$, produces a $(1+\frac{\kappa}{\log |V(G)|})$-approximation of the magnetization $\Mc_{G;\beta,\gamma,\lambda}$. We will show that we can approximate $\MaxCut$ on 3-regular graphs within a constant factor arbitrarily close to 1, contradicting the inapproximability result of \cite{maxcut}.

 Let $A(R),B(R), C(R)$ be the rational functions in Lemma~\ref{lem:maingadget} and set $D(R)=B(R)-C(R)$. From Footnote~\ref{ft:4f343}, we have that for all but finitely many values of $R$ we have that $D(R)\neq 0$. Using Theorem~\ref{thm:occgadget}, we therefore conclude that there are constants $R^*,M, \Xi, L>0$ such that $D(R^*)\neq 0$ and an algorithm, which, on input a rational $r\in (0,1/2)$, outputs in time  $poly(\bit(r))$ a pair of field gadgets  $\mathcal{T}_1,\mathcal{T}_2$, each of maximum degree $3$ and size $\leq L|\log r|$, such that 
\begin{equation}\label{eq:fieldgadgets}
|R_1 - R^*|,|R_2 - R^*|\leq r, \quad |M_1|,|M_2|\leq \Xi, \quad |M_1-M_2|\geq \hat{M}.
\end{equation}
where, for $i\in\{1,2\}$, $R_i:=R_{\Tc_i}$ is the effective field of $\Tc_i$ and $M_i:=M_{\Tc_i}$ is the magnetization gap of $\Tc_i$. For convenience, for $i\in\{1,2\}$ we let $\Ac_i',A_i,B_i,C_i,D_i$ denote $\Ac_{\Tc_i}',A(R_i),B(R_i),C(R_i), D(R_i)$, respectively.  Since the functions $A,B,C,D$ are continuously differentiable with respect to $R$, $D(R^*)\neq 0$ and $A(R^*)>1$, there are small constants $r_0,\zeta>0$ and large constants $\eta, D^*>1$ so that, for all $r<r_0$,   \eqref{eq:fieldgadgets} further implies that
\begin{equation}\label{eq:conti123}
\begin{gathered}
A_1,A_2\geq 1+\zeta, \quad D_1,D_2\neq 0 \mbox{ and } |D_1|,|D_2|\leq D^*,\\
\Big|\frac{R_1}{R_2}-1\Big|\leq \eta r, \quad |C_1-C_2|\leq \eta r, \quad |D_1-D_2|\leq \eta r, \quad \Big|\frac{1}{\log  A_1}-\frac{1}{\log  A_2}\Big|\leq \eta r.
\end{gathered}
\end{equation}
Let $k$ be a large integer satisfying $k>1+10/\log (1+\zeta)$ and $\epsilon>0$ be a small constant (both to be specified later). Let $G\in \Gc^{3k}_n$ be a graph satisfying Lemma~\ref{lem:SlySun}.

Let $H$ be a 3-regular graph, an instance of the $\MaxCut$ problem. We will use the algorithm of Theorem~\ref{thm:occgadget} for $r=\min\{\tfrac{\epsilon}{50\eta k\Xi LD^*|V(G)||V(H)|^3},r_0\}$, so we obtain field gadgets $\Tc_1,\Tc_2$ that satisfy \eqref{eq:fieldgadgets} and therefore \eqref{eq:conti123} as well.  For $i\in \{1,2\}$, consider the graph $H^{k}_{G,\Tc_i}$ as in Lemma~\ref{lem:maingadget}. For later use, note that the number of vertices in $H^{k}_{G,\Tc_i}$ is at most 
\[|V(H)|\, |V(G)|+2k|E(H)|\, |V(\Tc_i)|\leq 10kL|E(H)|\log |V(H)|,\] 
using that $|V(\Tc_i)|\leq 4L\log |V(H)|$ and $|V(G)|\leq \log |V(H)|$ for all sufficiently large $|V(H)|$.

Let $\mu_i$ denote the Gibbs distribution on $H^{k}_{G,\Tc_i}$ with parameters $\beta,\gamma,\lambda$. By Lemma~\ref{lem:maingadget}, the average magnetization of the graph $H^{k}_{G,\Tc_i}$ satisfies 
\begin{equation}\label{eq:magsi}
\begin{aligned}
\Mc_{\beta,\gamma,\lambda}(H^{k}_{G,\Tc_i})&=4k\mathcal{A}_i'|E(H)|+\Eb_{\sigma\sim \mu_i}\Big[\big|\sigma_{V(\widehat{H}^{k}_{G})}\big|\Big]+(1\pm 8\epsilon)kM_{i} Q_i.
\end{aligned}
\end{equation}
where $Q_i=D_i\AvgCut_{\mu_i}(H)+C_i |E(H)|$ and $\AvgCut_{\mu_i}(H)$ satisfies
\begin{equation}\label{eq:avgcutib}
1/K_i\leq \displaystyle \frac{\AvgCut_{\mu_i}(H)}{\MaxCut(H)}\leq 1\mbox{ for }K_i:=1+\displaystyle\frac{6}{k\log A_i}.
\end{equation}
Let 
\begin{gather*}
\Dc:=\Mc_{\beta,\gamma,\lambda}(H^{k}_{G,\Tc_1})- \Mc_{\beta,\gamma,\lambda}(H^{k}_{G,\Tc_2}),\quad  \Dc':=\Eb_{\sigma\sim \mu_1}\big[\big|\sigma_{V(\widehat{H}^{k}_{G})}\big|\big]-\Eb_{\sigma\sim \mu_2}\big[\big|\sigma_{V(\widehat{H}^{k}_{G})}\big|\big],
\end{gather*}
so that from \eqref{eq:magsi}
\begin{equation}\label{eq:r455f35f3bbba}
\Dc=4k(\Ac_1'-\Ac_2')|E(H)|+\Dc'+(1\pm 8\epsilon)k \big((M_1-M_2)Q_1+M_{2} (Q_1-Q_2)).
\end{equation}
We next show that for all sufficiently large $|V(H)|$ it holds that
\begin{equation}\label{eq:r455f35f3bbb}
\begin{gathered}
\Dc'\leq |V(H)|\, |V(G)| \cdot 2k|E(H)|\big(|V(\Tc_1)|+|V(\Tc_2)|\big)\eta r\leq \epsilon,\\ |M_2(Q_1-Q_2)|\leq \Xi(D^*+2)|E(H)|\eta r\leq \epsilon/(10k).
\end{gathered}
\end{equation}
The first inequality follows from applying Lemma~\ref{lem:perturb} to the graph $H^{k}_G$ (without the field gadgets). Namely, for $i\in \{1,2\}$, denote by $\nu_i$ the Gibbs distribution on $H^{k}_G$ where every vertex in $\widehat{H}^k_G$ has field $\lambda$ whereas every vertex in $V(H^{k}_G)\backslash V(\widehat{H}^k_G)$ has field $\lambda R_i$. Then, we have that 
\begin{equation}\label{eq:3dcevtvt5hyhuhu}
\Eb_{\sigma\sim \mu_i}\big[\big|\sigma_{V(\widehat{H}^{k}_{G})}\big|\big]=\Eb_{\sigma\sim \nu_i}\big[\big|\sigma_{V(\widehat{H}^{k}_{G})}\big|\big]
\end{equation}
and hence the inequality follows by applying Lemma~\ref{lem:perturb} to each of the vertices in $V(\widehat{H}^{k}_{G})$ and using the bound on $\big|\tfrac{R_1}{R_2}-1\big|\leq \eta r$ from \eqref{eq:conti123}. The second inequality in \eqref{eq:r455f35f3bbb}  follows from the fact that $|M_2|$ is bounded by $\Xi$, the bounds on $|C_1-C_2|, |D_1-D_2|, |D_1|, |D_2|$ from \eqref{eq:conti123} and the bound $\AvgCut_{\mu_i}(H)\leq \MaxCut(H)$ from \eqref{eq:avgcutib}.

Using the purported algorithm for the magnetizations on $H^{k}_{G,\Tc_i}$, and since the latter graph has at most $10kL|E(H)|\log |V(H)|$ vertices for all sufficiently large $|V(H)|$, we can compute an estimate of $\Mc_{\beta,\gamma,\lambda}(H^{k}_{G,\Tc_i})$ that is off by at most an additive 
\[\frac{\kappa}{\log |V(H^{k}_{G,\Tc_i})|}10kL|E(H)|\log |V(H)|\leq 10 k L |E(H)|\kappa.\] By subtracting these estimates for $i\in \{1,2\}$, we therefore compute $\widehat{\Dc}$ which satisfies
\begin{equation}\label{eq:r455f35f3bbbc}
|\widehat{\Dc}-\Dc|\leq 20 k L |E(H)|\kappa.
\end{equation}
Then, since the magnetization gaps $M_1,M_2$ satisfy $|M_1-M_2|\geq \hat{M}$ and $D_1\neq 0$ (see \eqref{eq:fieldgadgets} and \eqref{eq:conti123}), we can compute $\widehat{\textsc{MC}}=\frac{\widehat{\Dc}-4k(\Ac_1'-\Ac_2')|E(H)|}{k(M_1-M_2)D_1}-\frac{C_1}{D_1} |E(H)|$ which using  \eqref{eq:r455f35f3bbba}, \eqref{eq:r455f35f3bbb} and \eqref{eq:r455f35f3bbbc} satisfies 
\begin{equation}\label{eq:4g566g6b}
\begin{aligned}
|\widehat{\textsc{MC}}-\AvgCut_{\mu_1}(H)|&\leq \frac{20 k L |E(H)|\kappa+8 k |M_1-M_2|\, |Q_1|\epsilon+2\epsilon}{k\,|M_1-M_2|\, |D_1|}\\
&\leq \Big(\tfrac{20 L}{\hat{M}\, |D_1|}\kappa +8(1+\tfrac{|C_1|}{|D_1|})\epsilon+\tfrac{2}{\hat{M}\, |D_1|}\epsilon\Big)|E(H)|.
\end{aligned}
\end{equation}

By choosing $k$ sufficiently large, we have from \eqref{eq:avgcutib} that $\AvgCut_{\mu_1}(H)$ is within a factor  arbitrarily close to 1 from  $\MaxCut(H)$. By choosing $\epsilon, \kappa$ to be sufficiently small positive constants we can further ensure from \eqref{eq:4g566g6b} that our approximation $\widehat{\textsc{MC}}$ is within a factor  arbitrarily close to 1 from  $\AvgCut_{\mu_1}(H)$, and hence from $\MaxCut(H)$ as well. This finishes the contradiction argument and completes the proof of Theorem~\ref{thm:maingen2}.
\end{proof}

\begin{proof}[Proof of Theorems~\ref{thm:mainind2} and~\ref{thm:mainIs2}]
The theorems can  be obtained as immediate corollaries of Theorem~\ref{thm:maingen2} using the field-gadget constructions of Theorem~\ref{thm:const1}. As in Section~\ref{sec:f4343c}, we again note  that in the case of the hard-core model ($\beta=1,\gamma=0$), we have that $(\beta,\gamma,\lambda)\in \Uc^*_\Delta$ iff $\lambda>\lambda_c(\Delta)$, for all $\Delta\geq 3$. Similarly in the case of the antiferromagnetic Ising model ($0<\beta=\gamma<1$), we have that $(\beta,\gamma,\lambda)\in \Uc^*_\Delta$ iff $\lambda\in  (\tfrac{1}{\lambda_c},\lambda_c)$, where $\lambda_c=\lambda^{\Ising}_c(\Delta,\beta)$, for all $\Delta\geq 3$.
\end{proof}

\bibliographystyle{plain}
\bibliography{\jobname}

\end{document}